%% file: main.tex
\documentclass[11pt]{article}
\usepackage[utf8]{inputenc}
\usepackage{enumitem}
\usepackage{fullpage}
\usepackage{hyperref}
\hypersetup{
    colorlinks = true,
    citecolor = blue
}

\usepackage[skins,xparse,breakable]{tcolorbox}
\usepackage{amsmath,amsfonts,amssymb}
\usepackage{bbm}
\usepackage{circuitikz}
\usepackage{tikz}
\usepackage{standard}
\usepackage{multicol}

\usepackage{graphics}

\usepackage{dsfont}
\usepackage{amsthm}
\usepackage{sectsty}
\usepackage{gastex}
\newtheorem{corollary}[theorem]{Corollary}

\usepackage{hhline}

\newtheorem*{theorem*}{Theorem}

\usepackage{physics}
\usepackage{braket}

\usepackage{pgfplots}
 
\pgfplotsset{compat = newest}
\usepackage{authblk}

\usepackage{colortbl}

\begin{document}

\title{Quantum Communication Complexity of Classical Auctions}
\author{Aviad Rubinstein\footnote{aviad@cs.stanford.edu. Supported by NSFCCF-1954927, and a David and Lucile Packard Fellowship.}}
\author{Zixin Zhou\footnote{jackzhou@stanford.edu. Supported by NSFCCF-1954927, and a Stanford Interdisciplinary Graduate Fellowship.}}
\affil{Stanford University}

\maketitle
\begin{abstract}
    We study the fundamental, classical mechanism design problem of single-buyer multi-item Bayesian revenue-maximizing auctions under the lens of communication complexity between the buyer and the seller. Specifically, we ask whether using quantum communication can be more efficient than classical communication. We have two sets of results, revealing a surprisingly rich landscape --- which looks quite different from both quantum communication in non-strategic parties, and classical communication in mechanism design.

    We first study the expected communication complexity of approximately optimal auctions. We give quantum auction protocols for buyers with unit-demand or combinatorial valuations that obtain an arbitrarily good approximation of the optimal revenue while running in exponentially more efficient communication compared to classical approximately optimal auctions. However, these auctions come with the caveat that they may require the seller to charge exponentially large payments from a deviating buyer. We show that this caveat is necessary - we give an exponential lower bound on the product of the expected quantum communication and the maximum payment.

    We then study the worst-case communication complexity of exactly optimal auctions in an extremely simple setting: additive buyer's valuations over two items. We show the following separations:
    \begin{itemize}
        \item There exists a prior where the optimal classical auction protocol requires infinitely many bits, but a one-way message of 1 qubit and 2 classical bits suffices.
        \item There exists a prior where no finite one-way quantum auction protocol can obtain the optimal revenue. However, there is a barely-interactive revenue-optimal quantum auction protocol with the following simple structure:  the seller prepares a pair of qubits in the EPR state,  sends one of them to the buyer, and then the buyer sends 1 qubit and 2 classical bits.
        \item There exists a prior where no multi-round quantum auction protocol with a finite bound on communication complexity can obtain the optimal revenue. 
    \end{itemize}
\end{abstract}

\setcounter{page}{0}
\thispagestyle{empty}
\newpage

\input{newintro}

\section{Preliminaries I: Quantum}\label{sec:prelim-quantum}
\paragraph{Bra-ket notation}
In this paper, we may occasionally use bra-ket notation. Specifically, within an $N$-dimensional complex vector space, we represent each unit-length column vector as a ket, denoted as $\ket{\phi}$. Correspondingly, for every unit-length vector $\ket{\phi}$, a bra $\bra{\phi}$ is defined as an $N$-dimensional row vector that is the conjugate transpose of $\ket{\phi}$.

Moreover, we use the notation $\ket{a}$ for $a \in \{1, \ldots, N\}$ to indicate the column vector with a value of $1$ at the $a$-th coordinate and $0$ in all other positions. We refer to $\ket{1}, \ldots, \ket{N}$ as the computational basis.

We employ the notation $\ket{\phi}$ to represent a pure quantum state associated with the density matrix $\ket{\phi}\bra{\phi}$. Inversely, a quantum state described by the density matrix $\rho$ is considered a pure state if there exists a $\ket{\phi} \in \mathbb{C}^{N}$ such that $\rho = \ket{\phi}\bra{\phi}$.

\paragraph{Closeness of states}

Given two positive semidefinite matrices $\rho, \sigma \in \mathbb{C}^{N \times N}$, the trace distance between them is defined as 

\[
T(\rho, \sigma) = \max_{F} \frac{1}{2} \sum_{\ell} \left | \Tr(F_{\ell}\rho) - \Tr(F_{\ell}\sigma)  \right|,
\]
where $\{F_\ell \}$ is maximized over all possible POVMs\footnote{a positive operator-valued measure (POVM) is a finite set of positive semidefinite Hermitian matrices that sum to identity. }.

In particular, when $\rho$ and $\sigma$ are density matrices, $T(\rho, \sigma)$ is equal to the total variation distance between classical distributions obtained by measuring two states maximized over all possible measurements.

Below is an important property of the trace distance:

\begin{equation}
    \label{eqn:tracedistanceswaptest}
    T(\rho, \sigma) \le \sqrt{1 - \Tr(\rho\sigma)}.
\end{equation}

\subsection{(Non-Strategic) Quantum Communicatiom Protocols}

We give an overview of multi-party quantum communication protocols. For readers who are familiar with quantum communication, this model is equivalent to the ones used in the literature (e.g.~\cite{Yao93}). A formal description of two-party strategic communication model is given in Section~\ref{sec:model}.  A multi-party quantum communication protocol is defined over a system of qubits, that are initially partitioned between the parties.
The protocol proceeds in rounds; in each round, one party can locally manipulate or measure her qubits, and then send a subset of them to other parties. The {\em communication complexity} of a protocol is the total number of qubits sent by parties across all rounds.  We now provide more detail on each of those components.

\paragraph{Quantum systems}
Let $m$ be an upper bound on the number of qubits in the system.\footnote{We assume for simplicity of notation that there is a finite upper bound on the total number of qubits. Our results can be generalized e.g.~to a setting where each party can add qubits in each round of the protocol, and a setting where local operators are defined by general quantum channels.}
The state $\rho^{(r)}$ of the system at the beginning of round $r$ can be mathematically represented as a {\em density matrix}%
\footnote{A matrix is a {\em density matrix} if it is a positive semidefinite, trace $1$ Hermitian matrix. {\em Hermitian} means that $A^{\dagger} = A$, where $A^{\dagger}$ is the conjugate transpose of $A$.} $\rho^{(r)} \in (\mathbb{C}^{2 \times 2})^{\otimes m}$. 
Note that $(\mathbb{C}^{2 \times 2})^{\otimes m} = \mathbb{C}^{2^m \times 2^m}$, i.e.~it is just a $2^m$-by-$2^m$ complex matrix; however, the former tensor product notation will be useful when we consider the qubits held by each party.

\paragraph{Initial state of the system}
Initially, each party holds $m_i^{(0)}$ qubits ($\sum_i m_i^{(0)} = m$). 
Because we're concerned with quantum protocols for mechanisms with classical inputs, we assume that initially all the qubits are not entangled (e.g.~the initial state is $\rho^{(0)} = (\ket{0}\bra{0})^{\otimes m}$).
In particular, it is important that qubits held by different parties are initially non-entangled.

\paragraph{Local histories}
A party's local history consists of the number of qubits that she sent and received in each round so far in the protocol, as well as the outcomes of measurements that she locally performed on her qubits (see more on measurements below).

\paragraph{Local manipulations: unitary operators and measurements}
In each round, before sending any qubits, the active party can locally apply quantum operations and measurements to the qubits that she currently holds. If $m_i^{(r)}$ is the number of qubits held by party $i$ at the beginning of round $r$, we can represent the state $\rho^{(r)}$ as a density matrix in $(\mathbb{C}^{2 \times 2})^{\otimes m_i^{(r)}} \otimes  (\mathbb{C}^{2 \times 2})^{\otimes m-m_i^{(r)}}$. Party $i$'s operations can transform the state into 
\[ \rho^{(r+1/2)} = (U \otimes I_{2^{m-m_i^{(r)}}})^{\dagger} \rho^{(r)}(U \otimes I_{2^{m-m_i^{(r)}}}), \]
where $I_{2^{m-m_i^{(r)}}}$ is the identity operator on qubits held by other parties, $U$ is a unitary operator of $i$'s choice, acting only on $i$'s qubits, and $\rho^{(r+1/2)}$ is the new state of the quantum system after applying the operator (but before the measurement). 

Similarly, party $i$ can measure her qubits. 
A POVM  is defined by $L$ matrices $\{A_{\ell} \}_{\ell=1}^L \in \mathbb{C}^{2^{m_i^{(r)}} \times 2^{m_i^{(r)}}}$ such that $\sum_{\ell} A^{\dagger}_{\ell} A_{\ell} = I_{2^{m_i^{(r)}}}$.
After applying the measurement, with probability \[\Tr\left ( \left (A_\ell \otimes I_{2^{m-m_i^{(r)}}} \right)^\dagger \left (A_\ell \otimes I_{2^{m-m_i^{(r)}}} \right )\rho^{(r+1/2)} \right),\] the state of the system is updated to 
$$  \rho^{(r+1)} = \frac{\left(A_\ell\otimes I_{2^{m-m_i^{(r)}}}\right)^\dagger \rho^{(r+1/2)}\left(A_\ell \otimes I_{2^{m-m_i^{(r)}}}\right)} {\Tr \left (\left(A_\ell \otimes I_{2^{m-m_i^{(r)}}}\right)^\dagger \left (A_\ell \otimes I_{2^{m-m_i^{(r)}}} \right)\rho^{(r+1/2)} \right )}.   $$

It is wlog for each party to perform the measurement after all the unitary operators in a given round.

\paragraph{Sending and receiving qubits}
After applying local operations, the active party sends exactly one of the qubits she holds to another party.
Sending qubits does not change the state of the quantum system, but it changes which party can operate on the sent qubits. 

Note that it is wlog to send e.g.~the last qubit, because locally qubits can be swapped by unitary operators.

\paragraph{Termination of the protocol}
The protocol may terminate after a pre-determined number of rounds, or by any party as a function of her private inputs and/or local history. 

\paragraph{Complexity of the protocol}
The main metric of complexity of the protocol that we use is the total number of qubits sent by different parties. We give bounds for the complexity of both in-expectation (over the outcome of quantum measurements) and worst-case communication. In addition to the total number of qubits, we will show that:
(i) In some cases it is possible to simplify protocols by replacing some qubits with classical bits, and
(ii) we also consider the effect of restricting the number of rounds of the protocol.

\subsubsection{Conventions}\label{sub:simplifications}
We make the following conventions to simplify both our notation and analysis:
\begin{itemize}
    \item There is no seperate channel for classical information. For convenience, when we say a message is intended to be classical, it means the receiver immediately measures the qubit in the computational basis ($\ket{0}, \ket{1}$).
    \item  For convenience, in a 2-party protocol, when we mention that the  a player sends $m$ consecutive qubits to the other, it technically means that this player sends these qubit over $m$ rounds and the other player responds with a dummy qubit in each round.
\end{itemize}

\subsubsection{Choi-Jamiołkowski representation of protocols and strategies}
Consider a quantum protocol with a fixed number of rounds $R$ and a fixed number of qubits sent in each round.
A {\em strategy} $s_i$ of a party $i$ who is active in $R_i$ rounds is a sequence of $R_i$ mappings applied to the qubits that it holds at each round, together with a measurement of its qubits at the end of the protocol. A {\em co-strategy} $s_{-i}$ is a sequence of mappings by other parties at their rounds (and finally a measurement). Notice that the tuple of protocol, strategy, and co-strategy, fully determine the distribution of measurement realizations at the end of the protocol.

Suppose that party $i$'s measurements have $L_i$ possible outcomes and the other parties have $L_{-i}$ possible outcomes. \cite{GutoskiW07} show that any strategy (resp, co-strategy) can be represented as $L_i$ (resp. $L_{-i}$) matrices of dimension that depend only on the communication complexity, and not on the additional (possibly very large) quantum memory of the parties. 
For ease of presentation, we state only the simplest form of the theorem that we need; in particular, we avoid the notation necessary for actually defining the Choi-Jamiołkowski representation.

\begin{theorem}[Choi-Jamiołkowski representation of strategies in interactive quantum protocols~\cite{GutoskiW07}]
\label{thm:choi}
Consider any $R$-round quantum protocol with communication complexity $K$, a party $i$ in the protocol, and strategy $s_i$ for $i$ and co-strategy $s_{-i}$ for the rest of the parties. Let $\Phi^{(s_i)}, \Psi^{(s_{-i})}$ denote the respective Choi-Jamiołkowski representation.  Then the probability of measuring outcome $(a_i,a_{-i})$ is given by 
\[
2^K \cdot \Tr\left (\Phi^{(s_i)}_{a_i}  \Psi^{(s_{-i})}_{a_{-i}} \right).
\]
Moreover, each of $\Phi^{(s_i)}_{a_i}$, $\Psi^{(s_{-i})}_{a_{-i}}$ is a $2^{2K}$ by $2^{2K}$ positive semidefinite Hermitian matrix, and $\sum_{a_i = 1}^{L_i} \Tr (\Phi^{(s_i)}_{a_i} ) = \sum_{a_{-i} = 1}^{L_{-i}} \Tr (\Phi^{(s_{-i})}_{a_{-i}} ) = 1$\footnote{In the original definition of~\cite{GutoskiW07}, Choi-Jamiołkowski representations $\Phi^{(s_i)}, \Psi^{(s_{-i})}$ are not normalized. Here, we normalize all Choi-Jamiołkowski representations (now they are all density matrices), that is why we have an additional $2^K$ factor in the probability of outcome compared to the original paper. }.
\end{theorem}

\input{prelim-mechanism}

\input{newmodel-auction}

\section{Technical Overview}\label{sec:overview}
This section gives a brief and informal, high-level overview of our proofs.

\subsubsection*{Technical highlights: quantum auction protocol for combinatorial valuations} 
The starting point for this result is a mechanism $\cM$ that, for any buyer's type, samples an allocation from a distribution over $B$ possible subsets of items. As the first warm-up, we could let the buyer sample his allocation and specify it in $\log(B)$ bits; this is quite communication-efficient, but it is not IC: the buyer can always claim that he ``sampled'' his favorite subset $b^*$. 

Our second warm-up protocol is already quantum: we ask the buyer to encode the distribution $D(v)$ corresponding to his type $v$ as the quantum state $\ket{\psi_{D(v)}}=\sum_{b \in B} \sqrt{\Pr_{D(v)}[b]} \ket{b}$; this protocol uses $\log(B)$ qubits, and the seller can sample exactly from the correct distribution by measuring the state sent by the buyer. However, it is again not IC since the buyer can just send $\ket{b^*}$ for his favorite outcome $b^*$ --- the seller cannot distinguish between this and the intended message $\ket{\psi_{D(v)}}$ after she measures the latter. 

The strategic situation with the previous attempt is actually worse: even if the seller didn't already destroy the buyer's message by measuring it to sample the allocation, in general, there is no quantum measurement that can distinguish between (i) the set of valid $\ket{\psi_{D(v)}}$ that correspond to some lottery from the mechanism (aka for some type $v$), and (ii) the set of invalid $\ket{\psi_{D'(v)}}$ that don't correspond to an allowable lottery but may be preferable for the buyer.

Our key idea is to modify this non-IC quantum protocol by spot-checking the buyer: with low probability, we ask the buyer to resend the entire classical encoding of the same distribution (because this second stage happens with low probability, it has a negligible effect on the expected communication complexity). Given the classical encoding, the seller can (i) verify that the distribution is feasible; (ii) use the classical description to sample an outcome; and (iii) use a quantum measurement on the right basis to verify that the original quantum message was indeed close to the specific feasible distribution --- and penalize the buyer with a big payment if they fail this test.

\subsubsection*{Technical highlights: a lower bound on $(\text{CC}) \times (\text{payment})$}
This proof largely follows the framework of~\cite{RubinsteinZ21}: (i) reduce proving a lower bound on expected CC to worst-case CC; and (ii) use a counting argument to show that there aren't enough low-worst-case CC auction protocols to cover the large number of different priors that require different mechanisms. 

For (i), we simply have to truncate the protocol in the unlikely event that it significantly passes the expected communication. Because these events are very unlikely (by Markov's inequality), we expect the truncation to have a negligible effect on the expected buyer's incentives --- unless the protocol may charge the buyer a particularly high payment after an unlikely long interaction (notice that this is indeed what our efficient protocol does!). 

For (ii), a naive approach of counting low-CC protocols runs into the obstacle that the few qubits exchanged in the protocol can be entangled with each party's unbounded local memory in an arbitrarily complex way. 
We overcome this obstacle with a characterization due to~\cite{GutoskiW07} of quantum co-strategies that use the Choi-Jamiołkowski representation and are independent of the size of the local memory.

\subsubsection*{Technical highlights: worst-case quantum communication complexity of two-item auctions}
Our results in the second part of the paper build on works of~\cite{GK18,DaskalakisDT17} who characterize the revenue-maximizing auctions of various priors over independently, continuously distributed values $v_1,v_2 \in [0,1]$ for two items (assuming additive buyer's utilities). Specifically, they uniquely characterize the mechanism by the expected utility of the buyer $u(v_1,v_2)$ as a function of his values.

\paragraph{One-way quantum auction protocols}
Given a one-way quantum auction protocol we can write matrices $A^{(1)},A^{(2)},A^{(\text{pay})}$ that capture the seller's measurement of buyer's message such that the following hold: (i) Given a buyer with values $v_1,v_2$ for the items, his optimal strategy is to send a message that is an eigenvector of the matrix 
$$A(v_1,v_2) := (v_1 A^{(1)} + v_2 A^{(2)} - A^{(\text{pay})}),$$ 
and his expected utility is the maximum eigenvalue of the same matrix. So now we can equate the maximum eigenvalue of $A(v_1,v_2)$ with the expected utility function in the characterization of the optimal auction.
We observe the following:
\begin{itemize}
    \item A mechanism can be implemented by a finite classical mechanism iff the $u(\cdot,\cdot)$ is a piecewise linear function (with a finite number of pieces).
    \item If $u(\cdot,\cdot)$ has non-linear asymptotics, in particular, if it is a non-linear polynomial or rational function,  it cannot be equal to the maximum eigenvalue of $A(v_1,v_2)$.
\end{itemize}
We use an example from~\cite{DaskalakisDT17} with non-linear rational utility to show that some mechanisms cannot be implemented by finite one-way quantum protocols. We construct another example with utility that is non-linear but has the asymptotics of a linear function (specifically $u(v_1,v_2) =  v_1 - \frac{49}{24} + \frac14 (3 v_2 + \sqrt{121/4 - 10 v_2 + v_2^2})$) and show that it can be implemented by a one-way quantum protocol, but not by any finite classical protocol. 

\paragraph{Entanglement + one-way quantum auction protocols}
For the aforementioned example where the utility is a (non-linear) rational function and cannot be implemented by a one-way quantum auction protocol, we nevertheless give a finite interactive quantum auction protocol. Our protocol takes perhaps the simplest possible form for an interactive quantum protocol: the seller prepares an EPR pair and sends one of the qubits to the buyer; the buyer operates on his qubit from the EPR pair based on his private type and then sends it back to the seller. 

\paragraph{Limitations of finite quantum auction protocols}
We have seen that interactive multi-round quantum auction protocols are much richer than one-way quantum or finite classical protocols. 
However, even in this general case, the optimal strategy in any finite quantum protocol can be computed using an SDP~\cite{GutoskiW07}. In particular, the buyer's utility as a function of his values is always a semialgebraic function.
We use a construction due to~\cite{GK18} of a prior where the buyer's utility in the optimal mechanism is analytic but not semialgebraic to show that in general finite quantum auction protocols fail to achieve optimal revenue.

\part{Expected communication}\label{part:combinatorial}

\section{$\varepsilon$-IC Quantum Protocols for General Valuations}\label{sec:eps-IC-protocol}

Consider selling $n$ items to a single buyer with combinatorial valuations drawn from prior $\D$. We will show that for any direct IC mechanism $\mathcal{M}$ that only ever allocates $B$ different deterministic bundles, there is an $\varepsilon$-IC quantum protocol with the same expected payment for every type using $O(\log B)$ qubits of communication in expectation. Note that our protocol holds for arbitrarily small $\varepsilon$, and the constant factor in $O(\log B)$ does not depend on $\varepsilon$. 

Moreover, by employing a standard $\varepsilon$-IC-to-IC reduction that discounts all the payments by $(1-\sqrt{\varepsilon})$-factor, we can transform an $\varepsilon$-IC quantum protocol into an exactly IC quantum protocol that has the same expected communication complexity, while incurring only a negligible loss in revenue (see e.g.~\cite[Theorem 7]{GonczarowskiW18}).

\begin{theorem*}[Theorem~\ref{thm:combinatorial-protocol} restated] \hfill

    Let $\cal{D}$ be a prior over buyer's combinatorial valuations over $n$; assume all valuations are in the range $[0,U]$. 
    Let $\cM$ be any mechanism that can only possibly allocate one of $B$ subsets of the items. Finally, let $\delta > 0$ be any parameter ($\delta$ may be a function of $n$ or $B$). Then there is an IC auction quantum protocol that guarantees a $(1-\delta)$-fraction of $\cM$'s expected revenue using $O(\log(B))$ qubits in expectation.
\end{theorem*}

We first present a modification of an $O(B\log(B))$-bits classical protocol given in~\cite{RubinsteinZ21}. Then, we augment it with a single quantum message from the buyer to the seller that gives an exponential improvement in the expected communication. 

\subsection{Warm-up: An Inefficient Classical Protocol~\cite{RubinsteinZ21}}\label{sub:classical-protocol}

First note that if the maximum value of the buyer over any bundle is at most $U$, any IC and IR direct mechanism $\mathcal{M}$ that only ever allocates $B$ different bundles can be converted into an equivalent direct mechanism $\mathcal{M}'$ that only ever allocates $B$ different bundles and always charges either $0$ or $U$. This implies the new mechanism $\mathcal{M}'$ has $2B$ different deterministic outcomes (each deterministic outcome is a bundle-payment pair). Suppose $B$ different bundles allocated in $\mathcal{M}$ are $\pi_1, \ldots, \pi_B$, then $\mathcal{M}'$ creates $2$ outcomes $(\pi_j, 0), (\pi_j, U)$ for each bundle in $\mathcal{M}$. For each type, $v$, $\mathcal{M}$ has a distribution over $B$ bundles and an expected payment $P$ (by IR, $P \le U$). Then $\mathcal{M}'$ defines a distribution over $2B$ outcomes for each type $v$: suppose in $\mathcal{M}$ the proability of receiving $\pi_j$ is $p_j$, $\mathcal{M}'$ gives outcome $(\pi_j, U)$ with probability $p_j \cdot \frac{P}{U}$ and outcome $(\pi_j, 0)$ with probability $p_j \cdot (1 - \frac{P}{U}).$

We first present a classical (randomized) IC protocol that implements $\mathcal{M}'$ with $O(B \log B)$ expected communication complexity. This protocol differs slightly from the one in~\cite{RubinsteinZ21} and serves as a more straightforward foundation for constructing a quantum protocol. A probability distribution over $2B$ outcomes can be represented by $2B$ non-negative real numbers $p_1, \ldots, p_{2B}$ such that $\sum_{i=1}^{2B} p_j = 1.$  We call a distribution {\em feasible} if it is a distribution over $2B$ outcomes for some type in mechanism $\mathcal{M}'$. At each round, the buyer sends $2B$ bits, and the seller either terminates the protocol with an allocation and a payment or continues by moving to the next round and letting the buyer send more bits. 

\paragraph{Buyer's suggested strategy} Given the buyer's type and mechanism $\mathcal{M'}$, we denote by $p_1, \ldots, p_{2B}$ the probability of each outcome. The buyer sends $2B$ bits each round. The suggested strategy is to send, in the $r$-th Buyer round, the $r$-th bit of the binary representation of $p_1, \ldots, p_{2B}$, respectively.

\paragraph{Seller's strategy} After receiving each bit, the seller first checks if all buyer's messages so far are consistent with some feasible distributions, which means messages are binary prefixes of probabilities corresponding to some feasible distributions.  When there is only one possible value for the next bit that is consistent with some feasible distribution, then the seller sets the value of the next bit of the message to be the only feasible value and ignores the original bit of the message. 

Let $m^{(r)}_j$ be the $j$-th bit of message receives at round $r$. For round $r$, we denote by $\tau^{(r)}$ the total probability revealed so far, i.e. 
\[
   \tau^{(r)} = \sum_{t = 1}^r \sum_{j=1}^{2B} m^{(t)}_j \cdot 2^{-t}.
\]
In addition, we define $\tau^{(0)} = 0.$ After receiving all $2B$ bits of message at round $r$, the protocol terminates with probility $\frac{\tau^{(r)} - \tau^{(r-1)}}{1-\tau^{(r-1)}}$. Conditioned on terminating, the protocol assigns allocation and payment of outcome $j$ of $\mathcal{M}'$ with probability $ \frac{m_{j}^{(r)} \cdot 2^{-r}}{\tau^{(r)} - \tau^{(r-1)}}$ for each $j$.
Finally, with probability $1 - \frac{\tau^{(r)} - \tau^{(r-1)}}{1-\tau^{(r-1)}},$ the protocol continues. It is important to note that the protocol terminates in no longer than $r$ rounds with probability $\tau^{(r)}.$

\paragraph{IC} Since the protocol automatically keeps all messages consistent with a feasible distribution and since the buyer learns nothing about the seller's randomness other than that the protocol continues, it is easy to see that each strategy corresponds to some feasible distribution associated with a type (specifically, the distribution corresponding to the concatenation of all buyer's messages should the protocol continue indefinitely). Suppose strategy $s$ corresponds to a feasible distributions $\mathcal{D} = (p_1, \ldots, p_{2B}).$ We will show that for each outcome $i$, the probability of getting outcome $j$ when the buyer plays strategy $s$ is exactly $p_j$. Let 
\[ p_i = a_1\cdot 2^{-1} + a_2\cdot 2^{-2} + \cdots
\] be the binary representation of $p_j$.
Then the  the probability of realizing outcome $j$ when the buyer plays strategy $s$ is 
\[
 a_1\cdot 2^{-1} +  (1-\tau^{(1)}) \frac{a_2\cdot 2^{-2}}{1-\tau^{(1)}} + (1-\tau^{(2)}) \frac{a_3\cdot 2^{-3}}{1-\tau^{(2)}} + \cdots , 
\]
which is exactly equal to $p_j.$ Therefore, in this communication protocol the outcome distribution of a strategy associated with type $v$ is exactly the same as the outcome distribution of $v$ in mechanism $\mathcal{M}'$. Finally, by the fact that $\mathcal{M}'$ is IC, the suggested strategy is optimal for the buyer.

\paragraph{Communication complexity}

We will show that the protocol ends in $O(\log B)$ rounds in expectation. First note that the probability that the protocol does not end in $r$ rounds is $1 - \tau^{(r)} $. Next, by the fact that each strategy corresponds to a probability distribution $(p_1, \ldots, p_{2B})$, $1 - \tau^{(r)} = (\sum_{j}^{2B} p_j ) - \tau^{(r)} = \sum_{t=r + 1}^{\infty} \sum_{j=1}^{2B} m^{(r)}_j \cdot 2^{-t} \le 2B \cdot 2^{-r}.  $  As a result, the probability that the protocol does not end in $t$ rounds is at most $2B \cdot 2^{-r}.$  We can write the expected number of rounds $\mathbb{E} [T]$ as the following expression.
\begin{equation}
    \begin{aligned}
        \mathbb{E}[T] & = \sum_{t=1}^{\infty} \Pr[\text{protocol ends in $\ge t$ rounds}] \\
        & \le \sum_{t=0}^{\infty} 2\log B \cdot \Pr[\text{protocol ends in $> 2t\log B$ rounds}] \\
        & = 2 \log B + 2 \log B \cdot \sum_{k=1}^{\infty} \Pr[\text{protocol ends in $> 2k\log B$ rounds}] \\
        & \le 2 \log B + 2 \log B \cdot \sum_{k=1}^{\infty} 2B \cdot 2^{-2k \log B} \\
        & = 2 \log B + 4 \log B \cdot \sum_{k=1}^{\infty} \cdot B^{-2k + 1} \\
        & = O(\log B),
    \end{aligned}
\end{equation}
as desired.

Finally, the communication cost for each round is $O(B)$ bits. So the overall expected communication complexity is $O(B \log B)$.

\subsection{Quantum Protocol}
The idea of our quantum protocol is to replace all classical bits the buyer sends in the first $2 \log B$ rounds (total of $O(B \log B)$ bits) with a single message with $ \log(B) + 1  $ qubits and $2 \log B$ classical bits. Note that $ \log(B) + 1  $  qubits can encode an arbitrary distribution with support size $2B$ (see details below). The buyer is supposed to encode the distribution associated with his type into a quantum state. Then the seller can measure the quantum state and decide the allocation and payment according to the outcome. Here, the problem with this approach is that the buyer might encode an infeasible distribution (not associated with any type), and the seller has no way to tell if a quantum state encodes a feasible distribution. 

To overcome this challenge, the seller will,  with high probability, blindly trust the buyer and determine the allocation and payment according to the measurement outcome. However, with a small probability, rather than measuring the qubits, the seller asks the buyer to reveal the probability distribution he encoded by sending its full description classically. Subsequently, the seller can perform a measurement to verify if the quantum state accurately encodes the given distribution. Specifically, any quantum state that unfaithfully encodes the distribution will fail the test with a non-zero probability. As a result, the protocol can penalize the buyer with a big payment once she observes that the test has failed.

\paragraph{Seller's strategy}
In the first round, the seller expects a quantum message of $ (\log(B) + 1)  $ qubits, which we denote by $m_Q$, and a $\log (B+1)$-bit classical message that represents an integer $\hat{\tau} \in [0, B]$.

With probability $\gamma(\hat{\tau}) = 1 - \frac{1}{2B} - (1 - \frac{1}{2B})\frac{\hat{\tau}}{B^2}$, the seller measures the quantum message $m_Q$ in the computational basis $(\ket{1}, \ket{2}, \ldots, \ket{2B})$ and terminates the protocol. Suppose the measurement outcome is  $\ket{a} \in \{\ket{1}, \ket{2}, \ldots, \ket{2B}\},$ the allocation and payment are determined according to outcome $a$ of mechanism $\mathcal{M}'.$ 

With probability $ 1 - \gamma(\hat{\tau}) = \frac{1}{2B} + (1 - \frac{1}{2B})\frac{\hat{\tau}}{B^2}$, the seller asks the buyer to send $4B \log B$ classical bits to represent $2B$ binary numbers, $\widehat{p_1}, \ldots, \widehat{p_{2B}} $, each consisting of $2 \log B$ bits (we denote by $m_C$ this classical message). Applying the same correction procedure as described in Subsection~\ref{sub:classical-protocol}, we can assume that 
$\widehat{p_1}, \ldots, \widehat{p_{2B}} $ is the rounding-down to nearest multiple of $1/B^2$ of a feasible distribution 
$p_1, \ldots, p_{2B}$. 
The seller then verifies that $\hat{\tau} = B^2 \cdot (1 - \sum_{i} \widehat{p_i})$; if it isn't, the protocol terminates with the empty allocation and payment $2BU^3\varepsilon^{-2}$. Next, the seller measures the quantum message she receives earlier $m_Q$ in a way such that the measurement has two outcomes $0, 1$, and the probability of outcome $1$ is given by 
\[
\Tr(\rho \ket{\psi}\bra{\psi}),
\] where $\rho$ is the reduced density matrix that represents the state of $m_Q$ at the time of measurement, and $\ket{\psi}$ is the canonical state of classical message $m_C$ that is given by
\[
\ket{\psi} = \frac{B}{\sqrt{B^2 - \hat{\tau}}} \sum_{i=1}^{2B} \sqrt{\widehat{p_i}} \ket{i}.
\]

If the measurement outcome is $0$, the protocol terminates with empty allocation and payment $2BU^3\varepsilon^{-2}$. Otherwise, the protocol continues as a purely classical protocol in the following manner: for each $i \in \{ 1, \ldots, 2B\}$, with probability $\frac{1}{2B (1 - \gamma(\hat{\tau}))} \widehat{p_i}$, the protocol terminates with outcome $i$ of mechanism $\mathcal{M}'.$ Finally, with probability $1 -  \frac{1}{2B (1 - \gamma(\hat{\tau}))} \sum_i \widehat{p_i} $, the seller starts to run the classical protocol from round $2 \log B + 1$ and let $\widehat{p_1}, \ldots, \widehat{p_{2B}}$ be the message she received in the first $2 \log B$ rounds.

\paragraph{Buyer's suggested strategy}

Given the buyer's type and mechanism $\mathcal{M}',$ we denote by $p_1, \ldots, p_{2B}$ the probability of each outcome. Let $\widetilde{p_1}, \ldots, \widetilde{p_{2B}} $ be the numbers such that for any $i$, the binary representation of $\widetilde{p_i}$ consists of first $2 \log B$ bits of the binary representation of $p_i$. The suggested strategy is to send, for quantum message $m_Q$ whose density matrix is given by $\ket{\varphi}\bra{\varphi}$, where $\ket{\varphi}$ is defined as 
\[
\ket{\varphi} = \frac{B}{\sqrt{B^2 - \tilde{\tau}}} \sum_{i=1}^{2B} \sqrt{\widetilde{p_i}} \ket{i},
\]
and send integer $\tilde{\tau} = B^2(1 - \sum_i \widetilde{p_i}).$\footnote{Note that $ \widetilde{p_i}$ is a multiple of $2^{-2\log B} = \frac{1}{B^2}$, so $\tilde{\tau}$ is an integer. Moreover, $ 0 \le 1 - \sum_i \widetilde{p_i} = \sum_i (p_i - \widetilde{p_i})  \le B \cdot 2^{-2\log B} = \frac{1}{B}$, so $\tilde{\tau}\in [0, B].$ } In addition, when the protocol asks the buyer to send classical message $m_C$, the suggested strategy sends $\widetilde{p_1}, \ldots, \widetilde{p_{2B}}.$ Finally, as for the classical protocol part, for round corresponding to the $r$-th round of the classical protocol, the suggested strategy is to send the $r$-th bit of the binary representation of $p_1, \ldots, p_{2B}$, respectively.

We will show that for every type, the probability of each outcome when executing the suggested strategy in the quantum protocol is precisely the same as the probability of that outcome in mechanism $\mathcal{M}'$. 

It is important to note that the suggested strategy passes the measurement test with a probability $1$ and never pays the punishment payment. Hence, the suggested strategy is essentially equivalent to the suggested strategy of the classical protocol. 

For completeness, we include the verbatim calculation of probabilities. For outcome $i$, when measure the state $\ket{\varphi}$, the probability of obataining measurement $i$ is $\frac{B^2}{B^2 - \tilde{\tau}} \widetilde{p_i}.$ Next, the probability of getting outcome $i$ in the first $2 \log B$ rounds of the classical protocol is $ \frac{\widetilde{p_i}}{2B (1 - \gamma(\tilde{\tau}))}$. The probability of getting outcome $i$ in rest of the classical protocol is $(1 -  \frac{1}{2B (1 - \gamma(\tilde{\tau}))} \sum_j \widetilde{p_j} )\frac{(p_i - \widetilde{p_i}) B^2}{\tilde{\tau}} $. Finally, we assemble the probabilities of obtaining outcome $i$, we have 
\begin{align*}
    \Pr[\text{outcome } \, i] & = \gamma(\tilde{\tau}) \cdot \frac{B^2}{B^2 - \tilde{\tau}} \widetilde{p_i} + (1 - \gamma(\tilde{\tau})) \cdot  \left ( \frac{ \widetilde{p_i}}{2B (1 - \gamma(\tilde{\tau}))} +\left (1 -  \frac{ \sum_j \widetilde{p_i}}{2B (1 - \gamma(\tilde{\tau}))} \right)\frac{(p_i - \widetilde{p_i}) B^2}{\tilde{\tau}} \right) \\
    & =  \gamma(\tilde{\tau}) \cdot \frac{B^2}{B^2 - \tilde{\tau}} \widetilde{p_i} + \frac{ \widetilde{p_i}}{2B } +\left (1  - \gamma(\tilde{\tau})-  \frac{ \sum_j \widetilde{p_i}}{2B} \right)\frac{(p_i - \widetilde{p_i}) B^2}{\tilde{\tau}}  \\
    & = \gamma(\tilde{\tau}) \cdot \frac{B^2}{B^2 - \tilde{\tau}} \widetilde{p_i} + \frac{ \widetilde{p_i}}{2B } + \left (1 - \gamma(\tilde{\tau})-  \frac{ B^2 - \tilde{\tau}}{2B^3} \right)\frac{(p_i - \widetilde{p_i}) B^2}{\tilde{\tau}}  \\
    & =  \left ( 1 - \frac{1}{2B} - (1 - \frac{1}{2B})\frac{\tilde{\tau}}{B^2} \right) \cdot \frac{B^2}{B^2 - \tilde{\tau}} \widetilde{p_i} + \frac{ \widetilde{p_i}}{2B } + \left (1  - \gamma(\tilde{\tau}) -  \frac{ B^2 - \tilde{\tau}}{2B^3} \right)\frac{(p_i - \widetilde{p_i}) B^2}{\tilde{\tau}}  \\
    & =\left ( 1 - \frac{1}{2B} - (1 - \frac{1}{2B})\frac{\tilde{\tau}}{B^2} \right) \cdot \frac{B^2}{B^2 - \tilde{\tau}} \widetilde{p_i} + \frac{ \widetilde{p_i}}{2B } + \frac{\tilde{\tau}}{B^2} \frac{(p_i - \widetilde{p_i}) B^2}{\tilde{\tau}}  \\
    & = \frac{2B - 1}{2B} \widetilde{p_i} + \frac{ \widetilde{p_i}}{2B }  + p_i - \widetilde{p_i} \\
    & = p_i,
\end{align*}
as desired.

\paragraph{$\varepsilon$-IC}

Given the buyer's type and mechanism $\mathcal{M}',$ we denote by $p_1, \ldots, p_{2B}$ the probability of each outcome. We also denote by $u$ his expected utility. Let $M = 2^{4B \log B}$, we can represent $m_C$ by an integer $\ell \in [M]$. We will show that for any strategy, the expected utility he can get from the quantum protocol is at most $u + \varepsilon$. Wlog, every buyer's strategy can be defined by tuple $(\rho, \widehat{\tau}, \ell, s_C)$, where $\rho$ is the reduced density matrix of $m_Q$ and $\ell \in \{1, \ldots, M\}$  represents $m_C$, which is the classical message that consists of $4 B \log B$ bits the protocol wants the buyer to send with some small probability, and $s_C$ is a strategy of the classical protocol starting from round $2 \log B$ that is consistent with $\ell$. 

It is important to note that the buyer gains no advantage by performing a local measurement between sending $m_Q$ and  $m_C$, as any such strategy is equivalent to the strategy in which the buyer performs the same local measurement at the very beginning before sending $m_Q$.

Let $\psi_\ell$ be the canonical state associated with $\ell$, the probability that a strategy fails the test is given by
\begin{equation}
    \begin{aligned}
    \Pr[\text{Fail}] & = (1 - \gamma(\widehat{\tau})) \cdot (1 - \Tr(\rho\ket{\psi_{\ell}}\bra{\psi_{\ell}}) \\ 
    & \ge \frac{1}{2B} \cdot T^2(\rho,\ket{\psi_{\ell}}\bra{\psi_{\ell}}) \\
    \end{aligned}
\end{equation}
where $T(\cdot, \cdot)$ is the trace distance.

Note that once the buyer fails the test, he has to pay $2BU^3\varepsilon^{-2}$, therefore, the expected payment is at least 
\[
U^3\varepsilon^{-2} \left (  T(\rho,\ket{\psi_{\ell}}\bra{\psi_{\ell}}) \right)^2.
\]
As a result, for any optimal strategy, the inequality $T(\rho,\ket{\psi_{\ell}}\bra{\psi_{\ell}}) \le \frac{\varepsilon}{U}$ holds. Otherwise, the expected payment would be greater than the maximum value $U$. Suppose the optimal strategy  $s^*$ is given by $ \rho^*, \tau^*$, $\ell^*$, and $s_C^*$, where $s_C^*$ is the optimal strategy for the classical protocol that is consistent with $\ell^*$.

Now we consider a strategy $s$ that sends quantum state $\ket{\psi_{\ell^*}} \bra{\psi_{\ell^*}}$, integer $\tau^*$, and when the seller asks more bits, strategy $s$ sends $\ell^*$ and plays classical strategy $s_C^*$ in the classical protocol. It is important to note that, $s$ is the suggested strategy for some type if and only if $\tau^*$ is consistent with the message labeled with $\ell^*$. Therefore, the expected utility of $s$ is upper bounded by the suggested strategy for some type.  By IC of mechanism $\mathcal{M}'$ and the fact that the suggested strategy of each type achieves the same utility as mechanism $\mathcal{M}'$, the expected utility of $s$ is upper bounded by the expected utility of the suggested strategy.

To conclude the proof, we will demonstrate that the expected utility of $s^*$ is more than the expected utility of $s$ by at most $\varepsilon$. Conditioning on the protocol does not terminate after the first round (i.e.~the seller requires classical message $m_C$), by definition of $s$, the expected utility of $s$ is at least as large as that of $s^*$. Therefore, we only need to compare the expected utility contributed by measuring the reported quantum state immediately. The total variation distance between the measurement outcome distributions of $s$ and $s^*$ is upper bounded by 
\begin{align*}
    T \left (\rho^*, \ket{\psi_{\ell^*}}\bra{\psi_{\ell^*}} \right) \le \frac{\varepsilon}{U}.
\end{align*}
Finally, as the maximum value is bounded by $U$, the expected utility of $s^*$ is at most $\varepsilon$ greater than the expected utility of $s$, as desired.

\paragraph{Communication complexity}

First note that protocol ends immediately after the buyer's first message with probability $\gamma(\widehat{\tau}) \ge 1 - \frac{3}{2B}$. In this case, the communication cost is $O(\log B)$. Otherwise, the seller requests an additional $4 B \log B$ bits and may require more bits when running the classical protocol. As we have shown in the analysis of the classical protocol, the expected communication cost of this part is $O(B \log B)$. Since the probability that the protocol does not end in one round is no more than $1 - \gamma(\widehat{\tau}) \le \frac{3}{2B}$, the overall expected communication complexity is $O(\log B)$.
 
\section{Lower Bounds for Quantum Protocols with Small Payments}\label{sec:small-payments}

Throughout this section, we normalize the valuations so that the upper bound on the highest value is $U=1$.

\begin{theorem}[Full version of Theorem~\ref{thm:large-payments-intro}] \hfill \label{thm:large-payments} 

    Let $n$ be the number of items, $P$ be the upper bound on the payments and in the quantum auction protocol, and $\hat{K}$ an upper bound on the expected communication complexity. Then we have the following lower bounds on $\hat{K} P$:
    \begin{itemize}
        \item For unit-demand valuations, any quantum auction protocol that obtains $\Omega(1)$-approximation to the optimal revenue must satisfy $\hat{K} P=\Omega(n)$. 
        \item For Gross-substitutes valuations, any quantum auction protocol that obtains $\Omega(1)$-approximation to the optimal revenue must satisfy $\hat{K} P= 2^{\Omega(n^{1/3})}$. 
        \item For XOS valuations, any quantum auction protocol that obtains $\Omega(1)$-approximation to the optimal revenue must satisfy $\hat{K} P= 2^{\Omega(n)}$. 
        \item For XOS valuations over independent items, any quantum auction protocol that obtains $4/5+\Omega(1)$-approximation to the optimal revenue must satisfy $\hat{K} P= 2^{\Omega(n)}$. 
    \end{itemize}
\end{theorem}

\begin{proof}
    We write the proof for unit-demand valuations; the proof for other valuations followed by the respective bullet points of Lemma~\ref{lemma:hard-priors}.
    By Lemma~\ref{lem:inifinitetofinitereduction}, we can convert any quantum auction protocol with {\em expected} communication $\hat{K}$ and upper bound $P$ on payments and values to an $\varepsilon$-IC and $\varepsilon$-IR  quantum auction protocol {\em worst-case} communication complexity $K=O(\hat{K}P/\varepsilon)$. 
    
    By Lemma~\ref{lemma:hard-priors} (which is due to~\cite{RubinsteinZ21}) we know that there exists a family of $2^{2^{\Omega(n)}}$ priors, such that no single mechanism can simultaneously obtain any constant approximation of the optimal revenue for $\omega(1)$ of the priors in the family.
    
    Finally, by Lemma~\ref{lem:worstcasecounting}, given upper bound $K$ on worst-case communication, upper bound $B$ on number of feasible bundles%
    \footnote{$B=n+1$ for unit-demand, and $B=2^n$ for gross-substitutes and XOS}, and upper bound $P$ on highest payment, for any constant $\tau > 0$, there exists a cover, aka a family of $2^{2^{O(K + \log B + \log \log P + \log n)}}$ mechanisms such that the revenue of any quantum auction protocol with worst-case upper bounds $K$ on communication and $P$ on payment can be $(1 - \tau)$-approximately recovered by a mechanism in the cover. Therefore unless $K + \log \log P  = \Omega(n)$, there are not enough mechanisms in the cover to obtain good revenue against all priors in the family. 
\end{proof}

\subsection{Worst-Case-to-Expected Communication Reduction}\label{sub:wc2e}
To simplify the proof, we first recall a few assumptions on quantum protocols we made that do not sacrifice generality. 

First, we assume all messages are quantum messages. If a message from the buyer to the seller is intended to be a classical message, the seller simply measures these qubits immediately to obtain the classical bits. Note that the buyer gains no advantage by sending non-classical information, as the buyer can always measure the qubits by himself before sending them to the seller.

Next, we assume each party always sends exactly one qubit in each round. When the protocol stipulates that the buyer sends $m$ consecutive qubits to the seller, it formally means that the buyer sends these qubits over $m$ rounds while the seller sends back one dummy qubit in each round.

Finally, for any finite quantum auction protocol, if the size of each message is fixed across the protocol, by the deferred measurement principle, we can always assume that all measurement happen at the very end of the protocol. In fact, this also implies that only the seller is required to perform a measurement, while the buyer does not need to conduct any measurement.

\begin{lemma}[Worst-case-to-expected communication reduction]
\label{lem:inifinitetofinitereduction} \hfill

Let $X$ be a type space, and $\mathcal{P}$ be an $\varepsilon$-IC and $\varepsilon$-IR quantum auction protocol with $\hat{K}$ expected communication. Let $P$ be an upper bound of its payment and values in $X$. We can transform  $\mathcal{P}$ into a $(2\varepsilon)$-IC and $(2\varepsilon)$-IR finite quantum auction protocol $\mathcal{P}'$ with $K=O(\hat{K}P/\varepsilon)$ worst-case communication complexity. Moreover, for every type, the expected payment of $\mathcal{P}'$ is no less than the expected payment of $\mathcal{P}.$
\end{lemma}

\begin{proof}
Suppose the expected communication complexity of $\mathcal{P}$ is $\hat{K}$. 
We construct a new auction protocol, $\mathcal{P}'$, by terminating the protocol with an empty allocation and payment $P$ when the total communication cost would exceed $4 \varepsilon^{-1} \hat{K}P$ if the protocol were to proceed to the next buyer's round. According to Markov's inequality, the original auction protocol incurs a communication cost of $\ge 4 \varepsilon^{-1} \hat{K}P$ with a probability of at most $\frac{\varepsilon}{4P}$.

As a result, when employing any strategy from $\mathcal{P}$ on the trimmed $\mathcal{P}'$, the expected utility and payment change by at most $P \cdot \frac{\varepsilon}{4P} = \frac{\varepsilon}{4}$. Therefore, by designating the suggested strategy of $\mathcal{P}'$ as the (trimmed) suggested strategy of $\mathcal{P}$, it becomes evident that the new protocol is $(2\varepsilon)$-IC and $(2\varepsilon)$-IR with worst-case communication cost $4 \varepsilon^{-1} \hat{K}P = O(\hat{K}P/\varepsilon)$ while incurring no loss in expected revenue for each type.
\end{proof}

\subsection{A List-Decodable Code of Bayesian Priors}
The lower bounds of~\cite{RubinsteinZ21} against approximately optimal classical auction protocols construct, for each valuation class (unit-demand, submodular, etc.), a family of Bayesian priors over valuations from this class. Each family has the following ``list-decodability'' guarantee: no mechanism can simultaneously obtain high revenue on a ``list'' of $\omega(1)$ priors from the family.
We now state the results from~\cite{RubinsteinZ21}, with different parameters for family size and approximability for different valuation classes.

\begin{lemma}[Family of hard priors~\cite{RubinsteinZ21}]\label{lemma:hard-priors}\hfill

    For each of the following combinations of valuation class $X$ over $n$ items, family size $\zeta$, and approximability factor $\gamma$, there exists a family of $\zeta$ Bayesian priors over valuations from $X$, and a small constant $\varepsilon > 0$ such that no single $\varepsilon$-IC and $\varepsilon$-IR mechanism can simultaneously obtain a $\gamma$-approximation of the optimal revenue from $\omega(1)$ distinct priors from the family.
    \begin{itemize}
        \item $X = $ unit-demand; $\zeta = 2^{2^{\Omega(n)}}$; $\gamma = $ arbitrarily small constant. 
        \item $X = $ gross-substitute; $\zeta = 2^{2^{2^{o(n^{1/3})}}}$; $\gamma = $ arbitrarily small constant. 
        \item $X = $ XOS; $\zeta = 2^{2^{2^{\Omega(n)}}}$; $\gamma = $ arbitrarily small constant. 
        \item $X = $ XOS over independent items; $\zeta = 2^{2^{2^{\Omega(n)}}}$; $\gamma = 4/5 + \tau $ for arbitrarily small constant $\tau$. 
    \end{itemize}
\end{lemma}

\subsection{Approximate Cover over Efficient Quantum Auction Protocols}\label{sub:cover}
Our next objective is to demonstrate the existence of a finite set of protocols capable of approximating any (almost) IC and (almost) IR quantum protocol with bounded worst-case communication complexity.  Formally speaking, we have the following lemma.

\begin{lemma}[Approximate cover over efficient quantum auction protocols]
\label{lem:worstcasecounting} \hfill

Let $X$ be a type space with $n$ items, $B$ an upper bound of number of feasible bundles, $P$ an upper bound of payments and values of types in $X$, $\varepsilon > 0$ as an approximation parameter,  and  $K$ an upper bound on the communication complexity, there exists a set $\cS = \cS(n,B,P,\varepsilon,K) $ of mechanisms such that the following hold:
\begin{itemize}
    \item Small cover: $|\cS| = 2^{2^{O(K + \log B + \log \log P + \log n)}}$.
    \item Mechanisms in $\cS$ are $2 \sqrt{\varepsilon}$-IC and $2 \sqrt{\varepsilon}$-IR.
    \item Approximate covering property: For any $\varepsilon$-IC and $\varepsilon$-IR quantum protocol $\mathcal{P}$  with a worst-case communication complexity bounded by $K$, let $\mathcal{M}_{\mathcal{P}} = (\mathcal{A}_{\mathcal{P}}, \mathcal{Q}_{\mathcal{P}})$ be the mechanism induced by $\mathcal{P}$. Then there exists a mechanism $\mathcal{M} = (\mathcal{A}, \mathcal{Q}) \in \mathcal{S}$ such that for every type $v \in X$,  we have
\[
\mathcal{Q}(v) \ge (1 - \sqrt{\varepsilon}) \mathcal{Q}_{\mathcal{P}}(v).
\]
\end{itemize}  
\end{lemma}

\begin{proof}
 Each quantum protocol with $K$ worst-case communication can be characterized by its Choi-Jamiołkowski representation, allocation and payment mapping, and suggested strategy $s$. We denote by $\Phi_1, \dots, \Phi_L \in \mathbb{C}^{2K \times 2K}$ the Choi-Jamiołkowski representation of the quantum auction protocol, where $L$ is the number of measurement outcomes; similarly we let $\Psi  \in \mathbb{C}^{2K \times 2K}$ denote the representation of strategy $s$. Let $a_{\ell}$ and $p_{\ell}$ be the allocation and payment associated with  $\Phi_{\ell}$; their probability is given by
\[
2^K \Tr(\Phi_{\ell} \Psi).
\]

Suppose $\mathcal{P}$ only ever allocates bundles in set $\mathcal{B} \subseteq \{ 0, 1\}^{[n]}$ with $|\mathcal{B}| \le B$. For each feasible bundle $a \in \mathcal{B}$, we define 
\[
\widehat{\Phi}_a = 2^K \sum_{\ell, a_\ell = a} \Phi_{\ell}.
\]
Similarly, we define
\[
Q = 2^K \sum_{\ell} p_\ell \Phi_{\ell}.
\]

Moreover, by Theorem~\ref{thm:choi}, $\Phi_{\ell}$ is positive semidefinite and $\Tr(\sum_{\ell} \Phi_{\ell}) = 1$. As a result,  $\| \widehat{\Phi}_a \|_{\mathrm{F}} \le 2^K$ for all $a \in \mathcal{B}$ and  $\| Q \|_{\mathrm{F}} \le P \cdot 2^K$, where $\| \cdot \|_{\mathrm{F}}$ is the Frobenius norm. Note that the Frobenius norm upper bounds the absolute value of each entry of the matrix.

It can be readily observed that for any type $v \in X$ and any strategy $\Psi$, the expected utility is given by $\left ( \sum_{a} \Tr(\Psi \widehat{\Phi}_a) v(a) \right) - \Tr(\Psi Q),$ while the expected payment amounts to $\Tr(\Psi Q)$.
Hence, $\mathcal{P}$ can be represented by the set $\mathcal{B}$ of feasible allocations, and $B + 1$ Hermitian matrices $\{\widehat{\Phi}_a \}_{a\in \cB}$ and $Q$. 

 Now we consider discretizing $\widehat{\Phi}_a$ by rounding-toward-zero each entry to a multiple of $B^{-1} P^{-1} \cdot 2^{-\beta K}$ for some large constant $\beta$ (we round each entry in both real and imaginary parts).  Similarly, we discretize $Q$ by rounding toward zero each entry to a multiple of $2^{-\beta K}$. 
  
 Note that for two matrices $\rho$ with the same discretization, for any $\sigma$, we have 
\begin{align*}
\left | \Tr(\rho \sigma) - \Tr(\rho' \sigma) \right | & \le \| \sigma \| \cdot \left \| \rho - \rho' \right \|_1.
\end{align*}

As a result for every strategy $\Psi$,   
\[
   \left | \Tr(\widehat{\Phi}_{a} \Psi)  -  \Tr(\widehat{\Phi}'_{a} \Psi) \right | = B^{-1} P^{-1} 2^{-\Omega(K)},
\]
and 
\[
   \left | \Tr(Q \Psi)  -  \Tr(Q' \Psi) \right | =  2^{-\Omega(K)}.
\]

Therefore, for any strategy, the expected utilities the buyer obtains from two auction protocols with the same discretization differ at most $\cdot  2^{-\Omega(K)}$. Similarly, the difference in expected payment is at most $2^{-\Omega(K)}$. As a result, we only need to pick one protocol to approximate all protocols with the same discretization. We call the set of representative protocols $\mathcal{S}_I$. Next, we multiply all the payments of the protocols in $\mathcal{S}_I$ by a factor of $(1 - \sqrt{2\varepsilon})$ and increase all the payments by $\sqrt{2\varepsilon}$. We finally let the suggested strategy be the optimal strategy of each type. Let $\mathcal{S}$ be the set of all protocols constructed in this manner. 

Finally, we conclude that for every quantum $\varepsilon$-IC and $\varepsilon$-IR protocol with communication bounded by $K$, we can find an IC and $2\sqrt{\varepsilon}$-IR protocol in $\mathcal{S}$ while incurring at most $2\sqrt{\varepsilon}$ fraction loss of revenue for each type.

Now we count the number of protocols $\mathcal{S}$. For each $a$, by the fact that $\| \widehat{\Phi}_a \|_{\mathrm{F}} \le 2^K$, the absolute value of each entry is at most $2^K$. Therefore, the number of different discretized representation for each $a$ is bounded by $(B\cdot P \cdot 2^{O(K)})^{2^{O(K)}} = 2^{2^{O(K + \log \log P + \log \log B )} } $. Similarly, the number of different discretized representations of $Q$ is bounded by $2^{2^{O(K + \log \log P)} }$. In addition, the number of different $\mathcal{B}$ is at most $2^{nB}$. Finally, the cardinality $|\mathcal{S}|$ is upper bounded by $(2^{2^{O(K + \log \log P + \log \log B)} } )^ {O(B)} \cdot 2^{nB} = 2^{2^{O(K + \log \log P + \log n + \log B)} }.$

\end{proof}

\part{Worst-case communication}\label{part:2-items}

\section{Preliminaries III: Optimal Mehchanisms for Selling 2 Items}\label{sec:prelim-2-items}

In this section, we introduce a special case of framework of~\cite{DaskalakisDT17, GK18} to characterize optimal auctions for selling two goods to a single additive buyer with independent valuations.  We will first give a picture of their framework and then discuss how to apply duality theory to show that infinite menu complexity (aka infinite worst-case classical communication complexity) is inevitable for the optimal mechanisms of some prior distributions.

For simplicity, we only consider the case where the buyer's type space $X$ is $[0, 1]^2$, and each coordinate represents the buyer's value of each good. We assume the prior distribution has a density function $f(x, y) = f_1(x) f_2(y)$. Further, we assume $f_1, f_2 $ are continuous and differentiable with bounded derivatives. 
 
It is well-known (see e.g.~\cite{RochetC98, ManelliV07}, that for any IC and IR mechanism the utility function $u: [0, 1]^2 \to \mathbb{R}$ is convex, non-negative, non-decreasing, and $1$-Lipschitz with respect to the $\ell_1$ norm. Also, given any $u: [0, 1]^2 \to \mathbb{R}$ with these conditions, the utility function uniquely\footnote{Up to measure zero.} defines an IC and IR mechanism $\mathcal{M} $ with allocation function $\mathcal{A}(v) = \nabla u(v)$ and payment function $\mathcal{Q}(v) = \mathcal{A}(v) \cdot v - u(v).$

At a high level (see Figure~\ref{fig:canonicalfig}), the mechanisms established in~\cite{DaskalakisDT17, GK18} partition $[0, 1]^2$ into $4$ regions $Z, \mathcal{A}, \mathcal{B}, \mathcal{W}$ induced by a convex area $Z$ defined by two concave functions $s_1, s_2$ and a straight line $x + y = P_{\mathrm{crit}}$ for some $P_{\mathrm{crit}} \in [0, 2]$. Moreover, ~\cite{DaskalakisDT17} calls $Z, \mathcal{A}, \mathcal{B}, \mathcal{W}$ {\em the canonical partition with respect to $Z$}, and $P_{\mathrm{crit}}$ {\em the critical price}.
For ease of exposition, we only consider the case that the line $x + y = P_{\mathrm{crit}}$ intersects both curves $x = s_1(y)$ and $y = s_2(x)$. 

More precisely, let $s_1, s_2 : [0, 1] \to [0, 1]$ be two $1$-Lipschitz, concave and non-increasing function, and $P_{\mathrm{crit}} \in [0, 2].$ Let $x_{\mathrm{crit}}$ be the solution to $s_2(x) = P_{\mathrm{crit}} - x$, and $y_{\mathrm{crit}}$  be the solution to $s_1(y) = P_{\mathrm{crit}} - y$. We can always find such solutions since we assume the line $x + y = P_{\mathrm{crit}}$ intersects both curves.

Region $Z \subseteq [0, 1]^2$ is defined as the region enclosed by $s_1, s_2, $ and line $x + y = P_{\mathrm{crit}}.$ Formally, $Z = \{(x, y) \in [0, 1]^2: x \le s_1(y), y \le s_2(x), x + y \le P_{\mathrm{crit}}\}.$ The other three regions are defined as follows: 

$$
\mathcal{A} = \{(x,y) \in [0, 1]^2 : x < x_{\textrm{crit}}\} \setminus Z;
\quad \mathcal{B} = \{(x,y) \in [0, 1]^2 : y < y_{\textrm{crit}}\} \setminus Z; \quad \mathcal{W} = [0, 1]^2 \setminus (Z \cup \mathcal{A} \cup \mathcal{B}).$$

\begin{figure}[!ht]
\begin{center}
    
\caption{The canonical partition}\label{fig:canonicalfig}
\begin{tikzpicture}
\begin{axis}[width=2.5in, height=2.5in,ymin=0, ymax=10, xmin=0, xmax=10,   ytick pos=left, ytick={3.56,0} , yticklabels={$y_{\textrm{crit}}$,$0$}, xtick={3,0},xticklabels={$x_{\textrm{crit}}$,$0$}, xtick pos=left]
  
  \addplot+[color=gray!50, fill=gray!50, domain=0:3,mark=none]
 {8-.16*x*x}
 \closedcycle;
   \addplot+[color=gray!50, fill=gray!50, domain=3:6,mark=none]
 {9.56-x}
 \closedcycle;
\addplot+[color=gray!50,fill=gray!50,domain=6:7.135,mark=none]
{4.56-(x-5)*(x-5)} \closedcycle;

  \addplot+[color=gray!90, fill=gray!50, domain=0:3,mark=none]
 {8-.16*x*x};
   \addplot+[color=gray!90, fill=gray!50, domain=3:6,mark=none]
 {9.56-x};
\addplot+[color=gray!90,fill=gray!50,domain=6:7.135,mark=none, solid]
{4.56-(x-5)*(x-5)};
 
 \addplot+[color=black, domain=0:10,samples=2,  mark=none, dotted]
 {3.56} 
 ;
 \addplot[color=black, , dotted, mark=none] coordinates{
(3,0)
(3,11)
};
\node at (axis cs:1.5,5){$Z$};
\node at (axis cs:4.5,2){$Z$};
\node at (axis cs:1.5,2){$Z$};
\node at (axis cs:4,4.5){$Z$};
\node at (axis cs:8.2,2){$\mathcal{B}$};    
\node at (axis cs:7.2,7.2){$\mathcal{W}$};
\node at (axis cs:1.5,8.5){$\mathcal{A}$};
\node (dest4) at (axis cs: 2.3, 7.1){};
\node (lab4) at (axis cs: 2.3,9.2){};
\node (dest5) at (axis cs: 6.9,0.7){};
\node (lab5) at (axis cs: 9.2,0.7){};
\draw [->] (lab4)--(dest4);
\draw [->] (lab5)--(dest5);
\node at (axis cs: 2.3,9.3){\small{$s_2$}};
\node at (axis cs: 9.3,0.7){\small{$s_1$}};
\end{axis}
\end{tikzpicture}

This figure is from~\cite{DaskalakisDT17} which visualizes the canonical partition. In region $Z$, the seller allocates nothing; in region $\mathcal{W}$, the seller allocates both items; in regions $\mathcal{A}$ and $\mathcal{B}$, the seller allocates items by a lottery.
\end{center}

\end{figure}

We now define a technical condition that we call {\em GK Conditions}, which is needed for the characterization in Theorem~\ref{thm:characterization}; it holds for our examples in Section~\ref{sec:1-way-protocol} and Section~\ref{sec:finite-round}. (It actually fails to hold for the example in Section~\ref{sec:limit-1-way} and Section~\ref{sec:barely-interactive}, but for that example, we can use a similar unique characterization result directly from~\cite{DaskalakisDT17}.)

\begin{definition}[GK conditions~\cite{GK18}]\hfill\label{def:GK}
    
Let $\mu(x, y) = 3f_1(x)f_2(y) + xf'_1(x)f_2(y) + yf'_2(y)f_1(x)$.  Given a canonical partition of $[0, 1]^2$ induced by $s_1, s_2$, and $P_{\mathrm{crit}}$, we say that it satisfies  {\em GK conditions with respect to $f_1, f_2$} if it satisfies the following conditions:
\begin{itemize}
        \item $\mu(x, y) > 0$ for all $(x, y) \in [0, 1]^2$, and
        \item $\int_{Z} \mu(x, y) = 1$, and
        \item $\int_{s_1(y)}^1 \mu(x, y) \dd x = f_1(1)f_2(y)$, for all $y \in [0, y_{\textrm{crit}}], $ and
        \item $\int_{s_2(x)}^1 \mu(x, y) \dd y = f_1(x)f_2(1)$, for all $x \in [0, x_{\textrm{crit}}]$.
\end{itemize}
\end{definition}

Finally by~\cite{GK18} (specifically, their Theorem $1$ and duality discussions in Section 2.2),  we have the following characterization of the optimal mechanism.

\begin{theorem}[Uniqueness and characterization of optimal auction~\cite{GK18}]\label{thm:characterization}

Given probability density functions $f_1, f_2$ over $[0, 1]$. Suppose that canonical partition  $(Z, \mathcal{A}, \mathcal{B}, \mathcal{W})$ induced by $1$-Lipschitz, concave and non-increasing functions $s_1, s_2$, and $P_{\mathrm{crit}} \in [0, 2]$ satisfies the GK conditions. The utility function $u(x, y)$ of the optimal IC and IR mechanism for selling two items to a single additive buyer with independent prior distributions $f_1, f_2$ is given by 
\[
u(x, y) = \max(0, x - s_1(y), y - s_2(x), x + y - P_{\mathrm{crit}}).
\]
Specifically, 
\begin{itemize}
	\item if $(x,y) \in Z$, $u(x, y) = 0$;
	\item if $(x,y) \in \mathcal{A}$, $u(x, y) = y - s_2(x)$;
	\item if $(x,y) \in \mathcal{B}$, $u(x, y) = x - s_1(y)$;
	\item if $(x,y) \in \mathcal{W}$, $u(x, y) = x + y - P_{\mathrm{crit}}$.
\end{itemize}
Moreover, the optimal utility function $u(x, y)$ is unique in regions $\mathcal{A}, \mathcal{B}, Z$. \footnote{This implies that the utility function of every optimal mechanism for prior $f$ must agree with $u$ in regions $\mathcal{A}, \mathcal{B}, Z$.}

\end{theorem}

 \section{Limitations of One-Way Quantum Protocols}\label{sec:limit-1-way}
 ~\cite{DaskalakisDT17} studies the optimal auction of selling two items to a single additive buyer with i.i.d. valuations from Beta$(1, 2)$. It characterizes the unique utility function $u(\cdot)$ for any optimal mechanism. 
 In particular they show that in the region $y=1$ and $x\in[0, 0.06]$, the unique utility function for optimal mechanism is given by $$u(x, 1) = \frac{2 - 2x}{4-5x}.$$ 
 
In this section, we are going to show that no finite one-way quantum mechanism can implement this utility function for $x\in [0, 0.06]$ and $y = 1$. It it worth noting that, although only one qubit is exchanged in a single round per definition of our main model, the following negative result applies to any one-way quantum protocol with an arbitrarily large (but finite) message size.

\begin{theorem}
    Given the two-item additive type space $[0, 1]^2$, for any non-linear rational function $R(x)$, for any $r > 0$, there is no IC quantum one-way protocol $\mathcal{P}$ has utility function $u(x, 1) = R(x)$ for $x \in [0, r).$
\end{theorem}

\begin{proof}

For ease of exposition, we present the proof for the specific example of $R(x) = \frac{2-2x}{4-5x}$. However, it will be clear that the same argument holds for any non-linear rational function.

Wlog, in a one-way quantum protocol the buyer sends a $k$-qubit state $\rho$ to the seller, then the seller measures the state with a POVM and allocates the bundle, and charges the buyer according to the measurement outcomes. Since we only care about the distribution of measurement outcomes, we can represent a POVM measurement by $L$ positive semidefinite Hermitians $A_1, \ldots, A_L$ such that $\sum_{\ell=1}^L \Tr(A_\ell) = 1.$ Let $a_\ell$ and $p_\ell$ be the allocation and the payment associated with $A_\ell$.

 Define $A_{\hat{a}} := \sum_{\ell: a_\ell = \hat{a}} A_\ell$, for $\hat{a} \in \{ 0, 1\}^2$. In this definition $A_{0,0}$ is the sum of POVM Hermitians such that the corresponding allocation is empty, and $A_{1,0}$ is the sum of Hermitians such that the corresponding allocation is item $1$, etc. With these definitions, when the buyer's value of the first item is $x$ and his value of the second item is $1$, we define 
\begin{align*}
	V(x) &= \sum_{\ell} v(a_\ell) A_\ell \\
	& = x\cdot (A_{1,0} + A_{1,1}) + (A_{0,1} + A_{1,1}).
\end{align*}
Note that for any non-negative $x$, $V(x)$ is a positive semidefinite Hermitian. Then we define the payment matrix
\begin{align*}
	Q &= \sum_{\ell} p_\ell A_\ell.
\end{align*}
The utility function $u(x)$ for the buyer with values $(x, 1)$ when he sends $\rho$ is given by 
\begin{align*}
	u(x) &= \Tr(\rho(V(x) - Q)).
 \end{align*}
Since $V(x)$ and $Q$ are both Hermitian, the optimal strategy for the buyer is to send the state $\rho = \ket{\varphi}\bra{\varphi}$, where $\ket{\varphi}$ is the eigenvector of the largest eigenvalue of $V(x) - Q$. Moreover, the utility is given by the largest eigenvalue of the Hermitian: $\| V(x)-Q \| = \| x \cdot (A_{1,0} + A_{1,1})+ (A_{0,1}+A_{1,1} - Q) \| = \| xA + B \|$, where $A := A_{1,0} + A_{1,1}$ and $B := A_{0,1}+A_{1,1} - Q$.

We prove by contradiction by assuming there is an IC quantum one-way protocol such that $u(x) = \frac{2-2x}{4-5x}$ for $x \in [0, c).$ We have established that for any $x\in [0, c)$, $u(x)$ is given by the largest eigenvalue of  $( xA + B )$ for some Hermitian matrices $A, B$ defined by the protocol.

Let 
\[\det \left (\lambda I -(xA +B) \right) = \sum_{i = 0}^{2^k} q_i(x) \cdot \lambda^i
\] 
be the characteristic polynomial of $xA + B$, where $q_i(x)$ is a polynomial in $x$ of degree at most $2^k - i$. In particular, by the definition of determinant, only one term (the product of all diagonal entries) of $\det \left (\lambda I -(xA +B) \right) $) has $2^k$-th power of $\lambda$, so $q_{2^k} = 1$.

Since $u(x)$ is an eigenvalue of $xA + B$, it has to be a root of the characteristic polynomial of $xA +B$.  
\[
\det \left (u(x) I -(xA +B) \right) = 0.
\] 

Moreover, as $u(x)$ is analytic on for $x \in [0, c)$, $\det \left (u(x) I -(xA +B) \right)$ is also analytic in $x$ in this range. We denote by $g(x)$  the derivative of $\det \left (u(x) I -(xA +B) \right)$, we have
\begin{equation}
\label{eqn:derivative}
\begin{aligned}
g(x) & := \dv{x} \det \left(u(x) I -(xA +B) \right) \\
& = \frac{\dd } {\dd x} \left (\sum_{i=0}^{2^k} q_i(x) \cdot u^i(x) \right)\\
& = \sum_{i = 0}^{2^k} \frac{\dd }{\dd x} \left(q_i(x)\cdot u^i(x)\right) \\
& = \sum_{i = 0}^{2^k}  \left ( q_i(x) \cdot i \cdot u^{i-1}(x) \cdot u'(x) \right ) + \sum_{i = 0}^{2^k}  \left ( q_i'(x) \cdot u^{i}(x)  \right ).
\end{aligned}
\end{equation}
Note that $g(x)$ should be zero everywhere on  $ [0, c)$ since it is the derivative of a zero function $\det \left(u(x) I -(xA +B) \right)$.

By plugging $u(x) = \frac{2-2x}{4-5x}$ into~(\ref{eqn:derivative}), we have 
\[
g(x) = \sum_{i = 0}^{2^k}  \left ( q_i(x) \cdot i \cdot \left (\frac{2-2x}{4-5x} \right )^{i-1}  \cdot  \frac{2}{(4-5x)^2} \right ) + \sum_{i = 0}^{2^k}  \left ( q_i'(x) \cdot \left (\frac{2-2x}{4-5x} \right )^{i}   \right ).
\]
Since $g(x)$ is analytic over $(-\infty , \frac45)$, by the identity theorem, $g(x) = 0$ for $x \in [0, c)$ implies that $g(x) = 0$ for $x\in (-\infty, \frac45)$. We are going to show that 
\[
\lim_{x \to \frac45 -} g(x) \cdot (4-5x)^{2^k + 1} \ne 0,
\]
which leads to a contradiction. 
\[
g(x) \cdot (4-5x)^{2^k + 1} =\sum_{i = 0}^{2^k}  \left ( q_i(x) \cdot 2i \cdot (2-2x)^{i-1}  \cdot  (4-5x)^{2^k - i} \right ) + \sum_{i = 0}^{2^k}  \left ( q_i'(x) \cdot \left (2-2x\right )^{i} \cdot (4-5x)^{2^k - i}   \right ).
\] 

Finally, by the fact that $q_{2^k}(x) = 1$, we conclude that 
\[
	\lim_{x \to \frac45 -} g(x) \cdot (4-5x)^{2^k + 1} = 2^{k+1}\cdot \left (\frac56 \right)^{2^k - 1} \ne 0.
\]

\end{proof}

\section{A Quantum One-Way Mechanism for An Uncountable Menu}\label{sec:1-way-protocol}

In this section, we are going to construct an example of a prior over two additive, independent items, where the corresponding optimal auction can: (i) be implemented by a one-way protocol in $1$ qubit and $2$ classical bits; but (ii) requires an uncountably infinite menu.

\begin{theorem*}[Theorem~\ref{thm:1-way-protocol} restated] \hfill

    For the problem of auctioning two items to a single buyer, there is a Bayesian prior over independent item values, such that there is a revenue-optimal one-way quantum auction protocol where the buyer sends 1 qubit and 2 classical bits; yet no finite classical auction protocol can achieve the optimal revenue.
\end{theorem*}

At a high level, we construct the example in the following steps:
\begin{enumerate}
    \item  First, we want to apply Theorem~\ref{thm:characterization} to identify the utility function for optimal mechanisms. To simplify the verification of GK conditions, we choose $f_1(x) = 1$, for $x \in [0, 1]$, aka the value for item 1 is drawn uniformly from $[0,1]$. By choosing $f_1(x) = 1$, the measure $\mu(x, y)$ defined in GK conditions can be simplified to $\mu(x, y) = 3f_2(y) + yf'_2(y).$
    \item Next, we choose a non-increasing, 1-Lipschitz, concave function $s_1$. In addition, we require $s_1$ to be a non-piecewise-linear function, so the utility function in the region $\mathcal{B}$ of the canonical partition $u(x, y) = x - s_1(y)$ is non-piecewise-linear, which implies no finite menu can characterize it. Moreover, as we discussed in the last section, to be able to be implemented by a one-way quantum protocol, $s_1(y)$ has to be a function in the form $\|Ay + B \|$ for some Hermitian matrices $A$ and $B$. After some trial and error, it turns out that $s_1(y) =  \frac{49}{24} - \frac14 (3 y + \sqrt{121/4 - 10 y + y^2})$ is a good idea. 
    \item With the chosen $s_1$, the next step is to reverse-engineer Theorem~\ref{thm:characterization} to obtain a probability density function $f_2(y)$, function $s_2$ and critical price $P_{\mathrm{crit}}$ such that the canonical partition induced by $s_1, s_2, P_{\mathrm{crit}}$ satisfies GK conditions. In particular, By the third bullet of Definition~\ref{def:GK}, $f_2(y)$ has to satisfy the following ODE:
    \[
          (1 - s_1(y)) \left ( 3f_2(y) + yf'_2(y) \right) = f_2(y).
    \]
    \item Finally, we construct a one-way quantum protocol whose utility function is exactly the one given by Theorem~\ref{thm:characterization}.
\end{enumerate}

 \begin{figure} \label{fig:monster}
\caption{Probability density function of $f_2(y)$}
\begin{center}

 \begin{tikzpicture}
 
\begin{axis}[
    xmin = 0, xmax = 1,
    ymin = 0, ymax = 4]

        \addplot[
        domain = 0:1,
        samples = 200,
        smooth,
        thick,
        black,
    ] {(6234.27*(5 + sqrt(5737) + 12*x - 6*sqrt(121 - 40*x + 4*x^2))^(-(75/58) - 7551/(58*sqrt(5737)))*(10 - 2*x + sqrt(121 - 40*x + 4*x^2))^3*(-5 + sqrt(5737) - 12*x + 6*sqrt(121 - 40*x + 4*x^2))^(-(75/58) + 7551/(58*sqrt(5737))))/(11 - 2*x + sqrt(121 - 40*x + 4*x^2))^(99/29)};
\end{axis}
\end{tikzpicture}

This figure is the PDF of $f_2$. Moreover, $f_2$ is convex, decreasing over $[0, 1].$
\end{center}
 \end{figure}

By solving the ODE in bullet $2$, we obtain

\begin{equation} \label{eqn:monster}
    f_2(y) = \frac{c \cdot \left ( -6 g(y) + 12y + \sqrt{5737} + 5 \right )^{-\frac{75}{58} - \frac{7551}{58 \sqrt{5737}}} \left ( g(y) - 2y + 10 \right)^3\left (6g(y)-12y + \sqrt{5737}-5 \right)^{\frac{7551}{58 \sqrt{5737}} - \frac{75}{58}}} {\left (11 - 2 y + g(y) \right)^{99/29}},
\end{equation}

where $g(y) = \sqrt{4y^2 - 40y + 121}$ , and $c$ is the normalization factor such that 
\[
\int_{0}^{1} f_2(y) \dd{y} = 1. 
\]

\subsection{Optimal Mechanisms}

In this subsection, we will further define another non-increasing, 1-Lipschitz, concave function $s_2$ and critical price $P_{\mathrm{crit}} \in [0, 2]$. Next, we verify the canonical partition induced by $s_1, s_2$, and $P_{\mathrm{crit}}$ satisfy GK conditions and give the characterization of the optimal mechanism by applying Theorem~\ref{thm:characterization}.

First, one can verify that $\mu(x, y) = 3f_2(y) + yf'_2(y)$ is negative for all $(x, y) \in [0, 1]^2$.

Next, given $f_1(x) = 1$, by the last bullet of Definition~\ref{def:GK}, we set $s_2(x) \approx 0.558$ as a constant function such that $\int_{s_2(x)}^1 (3f_2(y) + yf'_2(y)) \dd y = f_2(1)$. 

Finally, we set $P_{\mathrm{crit}} \approx 0.669$ such that $\int_{Z} (3f_2(y) + yf'_2(y)) = 1 $. Moreover, line $x + y = P_{\mathrm{crit}}$ intersects both curves $x = s_1(y)$ and $y = s_2(x)$. With $s_1, s_2$ and $P_{\mathrm{crit}}$ we now have $x_{\mathrm{crit}} \approx 0.111$ and   $y_{\mathrm{crit}} \approx 0.005$. By definition, the canonical partition induced by $s_1, s_2$ and $P_{\mathrm{crit}}$ satisfies GK conditions. Therefore, by Theorem~\ref{thm:characterization}, we give the following characterization of the optimal mechanism for selling two items to a single additive buyer with values independently distributed according to $f_1 = 1$ and $f_2$ given by (\ref{eqn:monster}).

 \begin{figure}

 \caption{The canonical partition induced by $s_1$, $s_2$, $P_{\mathrm{crit}}$.} \label{fig:partitionmonster}
 \begin{center}
      \begin{tikzpicture}
\begin{axis}[
	clip=false,
    xmin = 0, xmax = 1,
    ymin = 0, ymax = 1]
    
          \addplot[
        domain = 0.66412:0.666666,
        samples = 200,
        smooth,
        black,
        thick,
    ] {1/48*(117 - 72* x - sqrt(3193 + 1968* x + 576* x^2))};
    \addplot[
        domain = 0.111:0.66412,
        samples = 10,
        smooth,
        black,
    ] {0.669-x};

    \draw  (0.111,0.558) -- (0.111,1) node[anchor=south, font=\tiny] {$x_\textrm{crit} \approx 0.111$};
        \draw  (0.111,0.558) -- (0,0.558|-{axis cs:0,0.558}) node[anchor=east, font=\tiny] {$s_2 = 0.558$}  ;
         \draw  (0.66412, 0.00488) -- (1, 0.00488) node[anchor = west, font=\tiny] {$y_\textrm{crit} \approx 0.005$}  ;
         \node at (axis cs:0.2,0.2){$Z$};
         \node at (axis cs:0.05,0.8){$\mathcal{A}$};
         \node at (axis cs:0.6,0.6){$\mathcal{W}$};
         \node (nodeb) at (axis cs:0.8, 0.003){};
         \node (nodew) at (axis cs: 0.8, 0.2) {};
         \draw [->] (nodew)--(nodeb);
         \node at (axis cs: 0.8, 0.21){\small{$\mathcal{B}$}};

\node[anchor=north, font=\tiny] at (axis cs:0.666666,0) {$\frac23$};
\end{axis}
\end{tikzpicture}
 \end{center}
 \end{figure}

\begin{figure} \label{fig:zoomin}
        \caption{A close-up of the tripoint.}
    \begin{center}
 \begin{tikzpicture}
 
\begin{axis}[
	clip=false,
    xmin = 0.655, xmax = 0.675,
    ymin = 0, ymax = 0.02 ,
    tick label style={
    /pgf/number format/.cd,
    fixed,
    precision=3,
    /tikz/.cd
  }
    ]
    
      \addplot[
        domain = 0.66412:0.666666,
        samples = 200,
        smooth,
        thick,
    ] {1/48*(117 - 72* x - sqrt(3193 + 1968* x + 576* x^2))};
    \addplot[
        domain = 0.655:0.66412,
        samples = 10,
        smooth,
        black,
    ] {0.669-x};
 
\node at (axis cs:0.66,0.004){$Z$};
\node at (axis cs:0.67,0.003){$\mathcal{B}$};
\node at (axis cs:0.666,0.012){$\mathcal{W}$};

 \node at (axis cs:0.668,0.002){\small{$s_1$}};
          \draw [->] (0.667, 0.002) -- (0.666,0.002);

         \draw [dotted]  (0.66412, 0.00488) -- (0.675, 0.00488) node[anchor = west, font=\tiny] {$y_\textrm{crit} \approx 0.005$}  ;

\node[anchor=north, font=\tiny] at (axis cs:0.666666,0) {$\frac23$};
\end{axis}
\end{tikzpicture}
    
     This figure offers a close-up of the tripoint where regions $Z$, $\mathcal{B}$, and $\mathcal{W}$ meet, as seen in Figure~\ref{fig:partitionmonster}. Notice that $s_1$ is actually slightly curved even though it looks linear in the figure.
     \end{center}
\end{figure}

\begin{lemma} \label{lem:monster}
    The optimal mechanism for selling two items to a single additive buyer with values independently distributed according to $f_1 = 1$ and $f_2$ defined by (\ref{eqn:monster}) is given by the utility function 
    \[
    u(x, y) = \max(0, x - s_2(y), y - s_1(x), x + y - P_{\mathrm{crit}}).
    \]
    Specifically, $u(x, y) =  x - s_2(y)$ in region $B$ is not a piecewise linear function, and it is the unique utility function for all optimal mechanisms in region $B$.
\end{lemma}

\subsection{Exact One-Way Quantum Protocol}

In this subsection, we give an IC one-way quantum protocol with exactly the same utility as the one characterized in Lemma~\ref{lem:monster}. 
\paragraph{Protocol implementation}
In the first (and the only) round, the buyer send a single qubit with reduced density matrix $\rho$, and two classical bits $b_1, b_2$. If $b_1 = b_2 = 0$, then the seller terminates the protocol with empty allocation and payment $0$. If $b_1 = b_2 = 1$, then the seller terminates the protocol with allocation $\{ 1, 2\}$ and payment $P_{\mathrm{crit}} \approx 0.669$. If $b_1 = 0, b_2 = 1$, then the seller terminates the protocol with alloation $\{ 2\}$ and payment $s_2(0) \approx 0.558$. Finally, if $b_1 = 1, b_2 = 0$, the seller measures the qubit using the following POVM and terminates the protocol with allocation and payment associated with each measurement outcome.

\begin{itemize}
	\item \[ 
	A_1 = \begin{pmatrix}
2/5 & \sqrt{21}/16\\
\sqrt{21}/16 & 1/4 
\end{pmatrix}, 	a_{1} = \{ 1, 2\}, p_1 = 2.
      \]
    
    \item \[ 
	A_2 = \begin{pmatrix}
2/5 & -\sqrt{21}/16\\
-\sqrt{21}/16 & 1/4 
\end{pmatrix}, 	a_{2} = \{1, 2\}, p_2 = 0.
      \]
      \item \[ 
	A_3 = \begin{pmatrix}
1/5 & 0\\
0 & 0 
\end{pmatrix}, 	a_{3} = \{1, 2\}, p_3 = \frac{299}{24}.
      \]
      \item \[ 
	A_4 = \begin{pmatrix}
0 & 0\\
0 & 1/2 
\end{pmatrix}, 	a_{4} = \{1\}, p_4 = \frac{7}{12}.
      \]

\end{itemize}

 We define the buyer's suggested strategy as the optimal response to the seller's strategy, given his private value $x$ and $y$. Before we dive into the analysis of buyer's optimal strategy, it is worth noting that this protocol is automatically IC as the suggested buyer's strategy is exactly the optimal one. 

\paragraph{Buyer's optimal strategy}
First it is easy to see that for given the buyer's type $(x, y)$, his expected utility is $0$, $y - s_2(x)$, $x + y - P_{\mathrm{crit}}$ when sending $b_1 = b_2 = 0$, $b_1 = 0, b_2 = 1$, $b_1 = b_2 = 1$, respectively.
Next we consider the buyer's utility when he reports $b_1 = 1, b_2 = 0$ with type $(x, y)$,  we define
\begin{align*}
V(x, y) &= x \cdot (A_1+A_2 +A_3+A_4)+y\cdot (A_1 + A_2+A_3) \\
& = xI + \begin{pmatrix}
	y & 0\\
0 & y/2 
\end{pmatrix}.
\end{align*}
Similarly, we can define the payment matrix $Q$ as follows,
\begin{align*}
	Q & = p_1A_1 + p_2A_2+p_3A_3+p_4A_4 \\
	& = \begin{pmatrix}
	79/24 & \sqrt{21}/8\\
\sqrt{21}/8 & 19/24
\end{pmatrix}.
\end{align*}

Therefore, the optimal strategy is to send a state corresponding to an eigenvector of the largest eigenvalue of $V(x, y) - Q$, and the optimal utility is given by

\begin{align*}
    \| V(x, y) - Q \| & = \left \| xI + \begin{pmatrix}
	y & 0\\
0 & y/2 
\end{pmatrix} -  \begin{pmatrix}
	79/24 & \sqrt{21}/8\\
\sqrt{21}/8 & 19/24
\end{pmatrix}  \right  \| \\
& = x + \left \| \begin{pmatrix}
	y-79/24 & - \sqrt{21}/8\\
-\sqrt{21}/8 & y/2 - 19/24
\end{pmatrix} \right \| \\
& = x - \frac{49}{24} + \frac14 (3 y + \sqrt{121/4 - 10 y + y^2})\\
& = x - s_1(y). 
\end{align*}

Therefore, the utility function of the optimal strategy is given by 
\[
 \max(0, x - s_1(y), y - s_2(x), x + y - P_{\mathrm{crit}}),
\]
which is exactly the same as the one characterized in Lemma~\ref{lem:monster}, as desired.
 
\section{(Barely) interactive one-way quantum auction protocols}\label{sec:barely-interactive}

In Section~\ref{sec:limit-1-way} we see an example of a prior whose optimal mechanism cannot be implemented by a finite one-way quantum auction protocol. In this section, we introduce a barely interactive multi-round quantum auction protocol which is optimal for this example, i.e.  $f_1(x) = 2(1-x), f_2(y) = 2(1-y)$ for $(x, y) \in [0, 1]^2$.

\begin{theorem}[Theorem~\ref{thm:interactive-vs-1-way}]
For the problem of auctioning two items to a single buyer, there is a Bayesian prior over independent item values, such that there is a revenue-optimal quantum auction protocol where the seller sends 1 qubit to the buyer, who replies with 1 qubit and 2 classical bits; yet no finite classical or one-way quantum auction protocol can achieve the optimal revenue.
\end{theorem}

Due to~\cite{DaskalakisDT17} Section $8.2.1$, the (unique) optimal mechanism for this example can be characterized by the following lemma.

\begin{lemma}[\cite{DaskalakisDT17}]
\label{lem:ddtutility}
    The optimal mechanism for selling two items to a single additive buyer with prior $f_1(x) = 2(1-x), f_2(y) = 2(1-y)$ is given by the utility function\footnote{The expression looks slightly different from the one given in~\cite{DaskalakisDT17}, but they are identical on $[0, 1]^2$.}
\[
  u(x, y) = \max(0, x - s_1(y), y - s_2(x), x + y - P_{\mathrm{crit}}),
\]
where $P_{\mathrm{crit}} \approx 0.5535$ is the critical price defined in~\cite{DaskalakisDT17} Section $8.2.1$, and 
\[
s_1(y) =
\begin{cases}
    \frac{2-3y}{4-5y} \quad y \in [0, \frac{4}{15}] \\
    2 \quad y \in (\frac{4}{15}, 1] 
\end{cases},
\]

\[
s_2(x) =
\begin{cases}
    \frac{2-3x}{4-5x} \quad x \in [0, \frac{4}{15}] \\
    2 \quad x \in (\frac{4}{15}, 1] 
\end{cases}.
\]

\end{lemma}

\subsection*{Protocol implementation}

The seller's strategy is as follows.

In this protocol, the seller first prepares an EPR pair: $\frac{1}{\sqrt{2}}\left ( \ket{0}\ket{0} + \ket{1}\ket{1} \right)$ and sends one qubit of the EPR pair to the buyer.  

The density matrix of an EPR pair is 
\[
\rho_{\textrm{EPR}} = \begin{pmatrix}
\frac12 & 0 & 0 & \frac12  \\
0 & 0 & 0 & 0 \\
0 & 0 & 0 & 0 \\
\frac12 & 0 & 0 & \frac12 
\end{pmatrix}.
\]

Next, the seller receives one qubit from the buyer (and two classical bits). We denote by $b_1, b_2$ the two classical bits. If $b_1 = b_2 = 0$, the protocol terminates with empty allocation and payment $0$.  If $b_1 = b_2 = 1$  the protocol terminates with allocation $\{1, 2\}$ and payment $P_{\mathrm{crit}} \approx 0.5535$ defined in~\cite{DaskalakisDT17}. If $b_1 \ne b_2$, the seller measures the joint state (two qubits) of his half of the EPR pair and the qubit she receives from the buyer. 
The seller will use the following POVM and corresponding allocation and payments. For convenience, we define bundle $\pi = \{1\}$ if $b_1 = 1$, and  $\pi = \{2\}$ if $b_2 = 1$.

\begin{itemize}

\item \[ 
	A_1 = \begin{pmatrix}
\frac{3}{20} & 0 & 0 & \frac{1}{10}\\
0 & 0 & 0 & 0 \\
0 & 0 & \frac{2}{5} & 0 \\
\frac{1}{10} & 0 & 0 & \frac{2}{5}
\end{pmatrix}, a_{1} = \pi, p_1 = 3.
      \]
      	\item \[ 
	A_2 = \begin{pmatrix}
\frac{23}{80} & 0 & 0 & -\frac{1}{10}\\
0 & 1 & 0 & 0 \\
0 & 0 & \frac{3}{5} & 0\\
-\frac{1}{10} & 0 & 0 & \frac{3}{5}
\end{pmatrix}, 	a_{2} = \pi, p_2 = 0.
      \]
      \item \[ 
	A_3 = \begin{pmatrix}
\frac{9}{16} & 0 & 0 & 0 \\
0 & 0 & 0 & 0 \\
0 & 0 & 0 & 0 \\
0 & 0 & 0 & 0 
\end{pmatrix}, 	a_{3} = \{1, 2\}, p_3 = 0.
      \]
      
  \end{itemize}
  
 We can verify by calculation that this is a valid POVM as all three matrices are positive semidefinite and $A_1 + A_2 + A_3 = I.$ 

 We define the buyer's suggested strategy as the optimal response to the seller's strategy, given his private value $x$ and $y$. Before we dive into the analysis of buyer's optimal strategy, it is worth noting that this protocol is automatically IC as the suggested buyer's strategy is exactly the optimal one. 

\subsection*{Structure of buyer's optimal strategy}

Our first question is, what joint state of two qubits can be achieved by the buyer operating only on his qubit (before sending it back to the seller)?
 An obvious restriction is that the reduced density matrix of the first qubit of the seller (her half of EPR pair) should be   \[
 \begin{pmatrix}
1/2 & 0\\
0 & 1/2 
\end{pmatrix}, 
\] which is exactly the reduced density matrix of the first qubit of an EPR pair, since the buyer has no way to touch the seller's half. Furthermore, by the unitary equivalence of purifications, one can see that the joint state the seller has in the end can be any two-qubit  state that satisfies this condition. 
In other words, the buyer can control the joint state $\rho$ that the seller measures as long as it satisfies the condition $\Tr_{2} (\rho) = \frac12 I$, where $\Tr_2(\cdot)$ is the partial trace operator that traces out the second qubit from the system and get the reduced density matrix of the first qubit.

Suppose 
\[
\rho = \begin{pmatrix}
\rho_1 & \rho_2 & \rho_3 & \rho_4 \\
\rho_5 & \rho_6 & \rho_7 & \rho_8 \\
\rho_9 & \rho_{10} & \rho_{11} & \rho_{12} \\
\rho_{13} & \rho_{14} & \rho_{15} & \rho_{16} 
\end{pmatrix},
\] then 
\[
\Tr_2(\rho) = \begin{pmatrix} 
\rho_1 + \rho_6 & \rho_3 + \rho_8 \\
\rho_9 + \rho_{14} & \rho_{11} + \rho_{16}
\end{pmatrix}.
\]

Wlog, we only consider the optimal utility the buyer can obtain conditioned on he sends $b_1 = 1, b_2 = 0$.

The optimization problem that the buyer with type $(x, y)$ want to solve is the following:
 
 \begin{align*}
 & \max_{\rho \succeq 0, \Tr_2(\rho) = I/2}   \left \langle \rho, x \cdot \left ( F_1 + F_2 + F_3  \right) + y \cdot F_3 - 3\cdot F_1  \right \rangle \\
  =  & \max_{\rho \succeq 0, \Tr_2(\rho) = I/2} \left \langle \rho, x \cdot I + y \cdot F_3 - 3\cdot F_1  \right \rangle \\
  =  x + & \max_{\rho \succeq 0, \Tr_2(\rho) = I/2} \left \langle \rho, C(y)  \right \rangle,
 \end{align*}
where,
\[
 C(y)  = y \cdot F_3 - 3\cdot F_1 =\begin{pmatrix}
\frac{9}{16} y  -\frac{9}{20} & 0 & 0 & -\frac{3}{10}\\
0 & 0 & 0 & 0 \\
0 & 0 & -\frac{6}{5} & 0 \\
-\frac{3}{10} & 0 & 0 & -\frac{6}{5}
\end{pmatrix}.
\]

As you can see, non-zero entries appear only on the diagonal and corners. We are going to who that for any $y \in [0, 1]$,  for the optimization problem \[ \max_{\rho \succeq 0, \Tr_2(\rho) = I/2}  \langle \rho,  C(y) \rangle \] it is without loss of generality to only consider a restricted solution space where non-zero entries appear only on the diagonal and corners of $\rho$. 

\begin{claim}
\label{claim:diagonalsol}
Suppose $\rho'$ is the matrix after we removing all but diagonal or corner entries from $\rho$. If $\rho$ is a feasible solution to the original optimization, so is $\rho'$ and $\langle \rho, C(y) \rangle = \langle \rho', C(y) \rangle .$
\end{claim}

\begin{proof}

\begin{itemize}
\item First, removing non-diagonal values won't affect the partial trace condition $\Tr_2(\rho')  =  I/2.$ 
\item Second, let $a, e$ be the top-left and bottom-right entry of $\rho$ and $b, \bar{b}$ be the top-right and bottom-left entry of $\rho$. $\rho \succeq 0$ implies that all diagonal entries are non-negative reals and $a e - b \bar{b} \ge 0$.  Finally, observe that since $\rho'$ only has diagonal and corner entries, $\rho' \succeq 0$ if and only if all diagonal entries are non-negative reals and $a e - b \bar{b} \ge 0$.
\item Third, $C(y)$ only has non-zero entries on the diagonal and corners, so $\langle \rho, C(y) \rangle = \langle \rho', C \rangle .$

\end{itemize}
\end{proof}

By this claim, we only need to consider solution $\rho$ in the following form
\[
\begin{pmatrix}
\alpha & 0 & 0 & \beta\\
0 & \gamma & 0 & 0 \\
0 & 0 & \delta & 0 \\
\bar{\beta} & 0 & 0 & \epsilon
\end{pmatrix},
\]
Since $\rho$ is a positive-semidefinite Hermitian with $\Tr_2(\rho) = I /2$, we have the following additional constraints for $\rho$.

\begin{itemize}
    \item $\alpha, \gamma, \delta, \epsilon$ are non-negative reals.
    \item $\alpha + \gamma = \delta + \epsilon = \frac12.$
    \item $\alpha \epsilon - |\beta|^2 \ge 0.$ 
\end{itemize}

Now let's observe the objective $\langle \rho,  C \rangle$. We have, 
\begin{align*}
\langle \rho,  C \rangle & = \alpha \cdot \left (\frac{9}{16} y  -\frac{9}{20} \right) - \frac{3}{10}(\beta + \bar{\beta}) - \frac{6}{5}(\delta + \epsilon) \\
& = \alpha \cdot \left (\frac{9}{16} y  -\frac{9}{20} \right) - \frac{3}{10}(\beta + \bar{\beta}) - \frac{3}{5} \\
& \le \alpha \cdot \left (\frac{9}{16} y  -\frac{9}{20} \right) + \frac{3}{5}|\beta| - \frac{3}{5} \\
& \le \alpha \cdot \left (\frac{9}{16} y  -\frac{9}{20} \right) + \frac{3}{5}\sqrt{\alpha\epsilon} - \frac{3}{5}. \\
& \le \alpha \cdot \left (\frac{9}{16} y  -\frac{9}{20} \right) + \frac{3}{5}\sqrt{\frac{\alpha}{2}} - \frac{3}{5}. \\
\end{align*}
The equality holds when $\beta = -\sqrt{\frac{\alpha}{2}}, \delta = 0, e = \frac12.$

Finally, it is easy to see that w.l.o.g., we only need to consider solution $\rho$ of the following form, 
\[
\begin{pmatrix}
\alpha & 0 & 0 & -\sqrt{\frac{\alpha}{2}}\\
0 & \frac12 - \alpha & 0 & 0 \\
0 & 0 & 0 & 0 \\
-\sqrt{\frac{\alpha}{2}} & 0 & 0 & \frac12
\end{pmatrix},
\]
where $0 \le \alpha \le 1/2$ is a real number.

As a consequence, the utility function $u$ of the optimal strategy conditioned on $b_1 = 1, b_2 = 0$ is given by the following optimization
 \begin{align*}
 u_1(x, y) & = x + \max_{\rho \succeq 0, \Tr_2(\rho) = I/2} \left \langle \rho, y \cdot F_3 - 3\cdot F_1  \right \rangle \\
 &= x +  \max_{0 \le \alpha \le \frac12} \left (\frac{9}{16} y - \frac{9}{20} \right ) \alpha + \frac{3}{5} \sqrt{\frac{a}{2}} - \frac{3}{5}.
 \end{align*}
 
By solving the optimization, we have
\[
u_1(x, y) = \begin{cases}
x + \frac{2 - 3y}{5y - 4}, 0 \le y \le \frac{4}{15} \\
x + \frac{9y}{32} - \frac{21}{40}, \frac{4}{15} \le y \le 1.
\end{cases}
\]

Symmetrically, the optimal utility conditioned on $b_1 = 0, b_2 = 1$ is given by 
\[
u_2(x, y) = \begin{cases}
y + \frac{2 - 3x}{5x - 4}, 0 \le x \le \frac{4}{15} \\
y + \frac{9x}{32} - \frac{21}{40}, \frac{4}{15} \le x \le 1.
\end{cases}
\]

Taking into account the cases where $b_1 = b_2 = 0$ and $b_1 = b_2 = 1$, we deduce that the optimal utility obtained by a buyer with type $(x, y)$ is represented by $\max(0, u_1(x, y), u_2(x, y), x + y - P_{\mathrm{crit}})$.

Furthermore, we note that for $y \in [\frac{4}{15}, 1]$, the inequality $x + \frac{9y}{32} - \frac{21}{40} \le x + y - P_{\mathrm{crit}}$ holds true, which implies that the utility function of the optimal strategy for this auction protocol  is identical to that of the optimal mechanism characterized by Lemma~\ref{lem:ddtutility}. By selecting the suggested strategy as the optimal one, we establish an IC and IR quantum protocol that precisely implements the optimal mechanism for this prior.

\section{Limitations of finite-round quantum protocols}\label{sec:finite-round}

In this section we give an example where no finite IC and IR protocol obtains optimal revenue. 

\begin{theorem*}[Theorem~\ref{thm:non-algebraic} restated]\hfill

    For the problem of auctioning two items to a single buyer, there is a Bayesian prior over independent item values, such that there is a revenue-optimal classical auction protocol that requires a constant number of bits {\em in exepctation}; yet no finite quantum auction protocol can achieve the optimal revenue.
\end{theorem*}

We will first review some definitions and results about semialgebraic functions and sets. 
We then proceed to show that for any IC  finite-round quantum protocol, the utility function $u(x, 1)$ after fixing $y = 1$ is the optimal value of an SDP parameterized by $x$ (with a slight abuse of notation, let $u(x) = u(x, 1)$). 
Next we argue that such $u(x)$ has to be semialgebraic. Finally, we will give an example of two item auctions where the utility function corresponding to the optimal mechanism is not semialgebraic. 

\subsection{Semialgebraic preliminaries}

Below are definitions of semialgebraic sets and semialgebraic functions. (See e.g.~\cite{coste2000} for reference.)
\begin{definition}[semialgebraic sets]
    A subset of $\mathbb{R}^n$ is semialgebraic if it can be represented as a finite union of sets of the form: 
    \[
    \{x \in \mathbb{R}^n : f(x) = 0, g_1(x) > 0, \ldots, g_m(x) > 0  \},
    \]
    where $f$ and $g_i$s are real polynomials in $x$.
\end{definition}

Moreover, the semialgebraic sets are closed under finite unions, intersections, complement and closure.

As a corollary, the boundary of a semialgebraic set is also semialgebraic, as the boundary of a set is the intersection of its closure and its complement's closure.

\begin{corollary}
\label{coro:boundary}
The boundary of a semialgebraic set is also semialgebraic.
\end{corollary}

\begin{definition}[Semialgebraic functions]
    A function $f: \mathbb{R}^n \to \mathbb{R}$ is semialgebraic if its graph $\{(x, y) \in \mathbb{R}^{n+1} : f(x) = y \}$  is a semialgebraic set.
\end{definition}

Following is an important property of semialgebraic functions.
\begin{lemma}
    \label{lem:semialgebraicfunction}
    Let $f: \mathbb{R} \to \mathbb{R}$ be a semialgebraic function, then there is a non-zero bivariate polynomial $q(X, Y) \ne 0$, such that $q(x, f(x)) = 0$ for all $x \in \mathbb{R}$. 
\end{lemma}

The following theorem about semialgebraic sets is particularly useful in our application.

\begin{theorem}[Tarski-Seidenberg]
\label{thm:tarski-seidenberg}
    Let $A \subseteq \mathbb{R}^{n + m}$ be a semialgebraic set. Then
    \[
    \{x \in \mathbb{R}^n : (x, y) \in A \text{ for some } y \in \mathbb{R}^m \}
    \]
    is semialgebraic.
\end{theorem}

\subsection{The utility function of finite round IC protocol is semialgebraic}

\begin{lemma} \label{lem:quantumsemialgebraic}
    Given an IC an IR finite-round quantum auction protocol, its utility function $u(x)$ is semialgebraic.
\end{lemma}

\begin{proof}

Just as in Subsection~\ref{sub:wc2e}, wlog we assume the length of each message of a finite round quantum protocol is fixed and measurements only happen at the very end of the protocol. By a semidefinite programming (SDP) characterization of feasible strategies in~\cite{GutoskiW07}, given an auction protocol $\mathcal{P}$, one can write $u(x)$ in the following optimization over real variables: 

\begin{equation}
\begin{aligned}
    u(x)  = \max_{\mathbf{z}} \quad C(x, \mathbf{z}) \text{ subject to } \mathbf{z} \in S,
\end{aligned}    
\end{equation}
where $\mathbf{z} $ is a finite-dimensional real vector, $C$ is a polynomial objctive function, and the constraint $\mathbf{z} \in S$ can be represented in finitely many polynomial equalities and inequalities. We say that $\mathbf{z}$ is {\em feasible} if $\mathbf{z} \in S$.

From the aforementioned discussion it is easy to see that the set 
\[
  \{(x, y, \mathbf{z}) : x \in \mathbb{R}, y \le C(x, \mathbf{z}),  \mathbf{z} \text{ is feasible} \}
\]
is semialgebraic.

By Theorem~\ref{thm:tarski-seidenberg}, we know set
\[
  A = \{(x, y) : x \in \mathbb{R}, y \le C(x, \mathbf{z}) \text{ for some feasible $\mathbf{z}$ }  \}
\] 
is semialgebraic. Note that $(x, y) \in A$ if and only if \[y \le \max_{\mathbf{z} \text { is feasible }} C(x, \mathbf{z}).\] Next, by Corollary~\ref{coro:boundary}, the boundary of $A:$ $\mathrm{bd}(A)$ is semialgebraic. Observe that  
\[ \mathrm{bd}(A) = \{(x, y) : x \in \mathbb{R}, y = \max_{\mathbf{z} \text{ is feasible} } C(x,\mathbf{z})  \}. \] Finally, we further restrict $x \in [0, 1]$, we know that the graph of $u(x)$
\[
\{(x, y) : x \in [0, 1], y = \max_{\mathbf{z} \text{ is feasible}} C(x,\mathbf{z}) \}
\] is also semialgebraic.  Thus, we conclude that $u(x)$ is semialgebraic. 

\end{proof}

\subsection{A mechanism with a non-semialgebraic utility function}

\cite{GK18} characterizes the optimal mechanism for selling two items to an additive buyer with i.i.d. priors $\frac{e^{-x}}{1-1/e}$ (i.e. $f_1(x)  =  \frac{e^{-x}}{1-1/e}$ and $f_2(y)  =  \frac{e^{-y}}{1-1/e}$). In particular, they show that the utility function of the optimal mechanism satisfies  $u(x, 1) = 2x + W(e^{1-x} (2-x)) - 1$ for $x \in [0, 0.1]$, where $W(\cdot)$ is the Lambert $W$ function\footnote{Defined as the inverse function of $f(w) = w\cdot e^w$. Moreover, by the Lagrange inversion theorem, $W$ is analytic everywhere on $(-1/e, \infty)$.}. Furthermore, by Theorem~\ref{thm:characterization}, we also know this utility function is unique in this region $(y = 1, x \in[0, 0.1])$. 

Next, we show that the unique utility function $g(x) = 2x + W(e^{1-x} (2-x)) - 1$ is not a semialgebraic function. Together with Lemma~\ref{lem:quantumsemialgebraic}, this implies that no finite quantum IC protocol achieves exactly optimal revenue (aka completing the proof of Theorem~\ref{thm:non-algebraic}).

\begin{lemma}
    Let $f: \mathbb{R} \to \mathbb{R}$ be a semialgebraic function. $f(x)$ cannot be equal to $g(x) = 2x + W(e^{1-x} (2-x)) - 1$ on $[0, r)$ for any $r > 0$.
\end{lemma} 

\begin{proof}

To this end, we will show that there cannot be a \emph{non-zero} bivariate polynomial $P$ such that $P(x, g(x)) = 0$ for all $x \in [0, r).$ Note: it is sufficient to show that there is no \emph{non-zero} bi-variate polynomial $P$ such that $P(x, v(x)) = 0$ for all $x \in [0, r)$, where $v(x) = W(e^{1-x} (2-x))$. This is because $g(x) = v(x) + 2x - 1$, assuming by contradiction $P(x, g(x)) = 0$ on $[0, r)$, through expanding this polynomial, we sure have another polynomial $P'(x, v(x)) = 0$ on $[0, r)$.

We prove it by contradiction. Assume there exists a non-zero polynomial  $P(x, v(x)) = 0$, and it can be written as 
\[
P(x, y) = a_n (x) y^n + a_{n-1}(x) y^{n-1} + \cdots + a_0(x) ,
\]
where coefficients $a_i(x)$ are real polynomials in $x$ . 

By our assumption, $P(x, v(x)) = 0$ for $x \in [0, r)$. Furthermore, $v(x)$ is analytic for $x \in [0, +\infty)$, and so is $P(x, v(x))$. Thus, by the identity theorem, $P(x, v(x)) = 0$ for $x \in [0, +\infty).$ Next observe that $v(x)$ is negative, monotone increasing for $x \in [3, +\infty)$. More importantly, $e^{1-x} (2-x)>v(x) > 2e^{1-x} (2-x)$ for $x \in [3, +\infty)$. Let $l$ be the smallest index such that $a_l(x)$ is non-zero, we define 
\[
F(x, y) = \frac{P(x, y)}{y^l} =  a_n (x) y^{n-l} + a_{n-1}(x) y^{n-l-1} + \cdots + a_l(x),
\]
and then we plug $y = v(x)$ into it, 
\[
F(x, v(x))  =  a_n (x) v^{n-l}(x) + a_{n-1}(x) v^{n-l-1}(x) + \cdots + a_l(x).
\]

When $x$ goes to infinity, $\lim_{x \to \infty} a_i (x) v^{i-l}(x) = 0$ for all $i > l$ as $v(x)$ is exponentially small. Next, $\lim_{x \to \infty} a_l(x) \ne 0 $ as $a_l(x)$ is a non-zero polynomial. Finally, we conclude that 
\[
\lim_{x \to \infty} F(x, v(x)) \ne 0,
\]
which contradicts with the assumption that  $P(x, v(x)) = 0$ for $x \in [0, +\infty).$

\end{proof}

\bibliographystyle{alpha} 
\newcommand{\etalchar}[1]{$^{#1}$}

\end{document}

%% file: newintro.tex
\section{Introduction}

We study the quantum communication complexity of classical problems with strategic constraints. The communication problems we study differ from traditional (``cooperative'') problems in communication complexity like Set Disjointness due to Incentive Compatibility (IC) constraints.
Informally speaking, one can think of the goal of a traditional communication problem is to design   a protocol for both parties that optimizes some common objective function. 

For instance, in Set Disjointness, the common goal is to maximize the probability of outputing the correct answer. 
By contrast, in strategic communication, each party optimizes a differenet objective; we seek protocols that are communication-efficient, and at the same time don't expect strategic parties to take actions that are mis-aligned with their objectives.

Specifically, we study the quantum communication complexity of a fundamental setting in mechanism design: a monopolistic, revenue-maximizing seller with Bayesian prior auctioning $n$ items to a single buyer; this setting has proved very attractive to researchers in theoretical computer science (e.g.~\cite{BabaioffILW20, ChenDOPSY15, BabaioffGN22, Gonczarowski18, GuoHZ19, RubinsteinZ21} and references therein).

It is known that even with seemingly benign Bayesian priors, revenue-optimality requires complex auctions, e.g.~ones that allow the buyer to choose among lotteries~\cite{Thanassoulis04,ManelliV10,BriestCKW10,Pavlov11,HartR15}. 
This realization has sparked a fruitful line of work on the simplicity-vs-complexity of  (approximately) optimal auctions. 
Understanding the tradeoffs of simplicity vs.~complexity requires formal definitions of complexity. 

Perhaps the most well-studied notion of complexity for this problem is the number of distinct lotteries offered to the buyer (``menu-size complexity''~\cite{HartN19}). 
The measure exactly characterizes the {\em deterministic communication} complexity of the interaction between a buyer and seller who both know the rules of the auction%
 \footnote{Specifically, $\text{deterministic-CC} = \log\big(\lceil \text{menu-size complexity}\rceil \big)$.}~\cite{BabaioffGN22}.

We study the communication complexity of auctions subject to an Incentive Compatibility (IC) constraint: it is crucial that a strategic buyer must not be able to gain from deviating from the protocol.
(As is standard in the literature, we assume that the seller is non-strategic and follows the protocol faithfully. See Section~\ref{sec:model} for formal definition, and e.g.~\cite{FadelS09,DobzinskiR21,RubinsteinZ21,RubinsteinSTWZ21} for further discussion.)
Under this IC constraint, \cite{RubinsteinZ21} show that it is possible to obtain dramatic savings in communication by considering (classical) {\em randomized communication}. The main question we ask in this work is whether {\em quantum communication} can be even more efficient than classical randomized communication complexity for this problem:

\begin{quote}
    Can quantum auction protocols achieve super-classical performance? 
\end{quote}

Specifically,~\cite{RubinsteinZ21} show that even though randomized auction protocols can be much more efficient than deterministic ones, they  still have limitations:
\begin{itemize}
    \item {\bf Worst-case vs expected CC barrier:} \cite{RubinsteinZ21} improve the communication complexity in expectation (over the randomness of the protocol), but the worst-case communication complexity of randomized protocols is still characterized by the menu-size complexity. 
    \item {\bf Combinatorial valuations barrier:} For buyers with combinatorial valuations over the items, \cite{RubinsteinZ21} prove exponential lower bounds even for the expected communication of randomized auction protocols. These lower bounds hold even for approximately optimal mechanisms, and even against restricted classes of valuations (e.g.~monotone submodular valuations). 
\end{itemize}
In this paper, we investigate to what extent quantum communication can break these classical barriers.

\subsection{Our Contributions}
We formalize a model of {\em quantum auction protocols} (Definition~\ref{def:quantumauction}) that extends the randomized auction protocols of~\cite{RubinsteinZ21} by allowing the buyer and seller to send, receive, operate on, and measure qubits. We then provide two sets of results, centered around the two barriers for classical auction protocols mentioned above.

\subsubsection{(Un)expected Quantum CC with Combinatorial Valuations}
Our first result is an exponential quantum speed-up for auction protocols. It is stated in a general form for a mechanism that chooses an allocation among $B$ possible allocations. Specifically, it gives a near-equivalent IC quantum auction protocol that uses only $O(\log(B))$ qubits --- matching the cost of naively encoding the allocation (without strategic considerations). 
\begin{theorem}[Efficient in-expectation quantum auction protocols]\label{thm:combinatorial-protocol} \hfill

    Let $\cal{D}$ be a prior over buyer's combinatorial valuations over $n$ items; assume all valuations are in the range%
    \footnote{Our protocol assumes that we're given some finite upper bound $U$ on valuations, but the communication complexity does not depend on $U$.} $[0,U]$. 
    Let $\cM$ be any mechanism that can only possibly allocate one of $B$ subsets of the items. Finally, let $\delta > 0$ be any parameter ($\delta$ may be a function of $n$ or $B$). Then there is an IC quantum auction protocol that guarantees a $(1-\delta)$-fraction of $\cM$'s expected revenue using $O(\log(B))$ qubits in expectation.
\end{theorem}
As a corollary, we only need $O(\log(n))$  qubits for unit-demand, or $O(n)$ qubits for arbitrary combinatorial valuations. On the contrary, \cite{RubinsteinZ21} shows that any randomized classical communication protocol requires at least $\Omega(n)$ bits for unit-demand, and  $\Omega(2^n)$ bits for combinatorial valuations.
To the best of our knowledge, this is the first exponential separation of quantum and classical communication in algorithmic game theory.

The positive result in Theorem~\ref{thm:combinatorial-protocol} has a caveat: the protocol may require large payments. Specifically, we need the ability to inflict a large penalty on buyers who deviate from the intended quantum strategy. However, the probability of catching deviating buyers may be exponentially small, so we use exponentially large payments. Even though buyers who follow the protocol can never be penalized, to be accountable for a potentially large payment the buyer may need significant collateral to participate in the auction. 
Our second result shows that unfortunately without exponentially large payments all the lower bounds from~\cite{RubinsteinZ21} extend to quantum communication.

\begin{theorem}[Efficient protocols require large payments - short version] \label{thm:large-payments-intro} \hfill

    Let $n$ be the number of items, $P$ be an upper bound on the payments in the quantum auction protocol (when the valuations are normalized to $[0,1]$), and $\hat{K}$ an upper bound on the expected communication complexity. Then for a buyer with combinatorial (XOS) valuations, any quantum auction protocol that obtains $\Omega(1)$-approximation to the optimal revenue must satisfy $\hat{K} P= 2^{\Omega(n)}$. See Theorem~\ref{thm:large-payments} for full statement and additional results.
\end{theorem}

Interestingly, we are not aware of any classical analogs of such tradeoffs between maximum payment and communication complexity. In particular, the exponential lower bounds against classical auction protocols in~\cite{RubinsteinZ21} hold even with arbitrarily large payments.

\subsubsection{Worst-Case Quantum CC with Two Items}
Our second set of results focuses on the particularly simple case of a buyer with additive valuations over only two items, and the Bayesian prior for these values is independent. In this case, classical protocols of~\cite{RubinsteinZ21} already achieve $O(1)$ expected communication, but their worst-case communication is infinite. Note that this is for exactly optimal mechanisms --- for approximately optimal mechanisms, worst-case can be reduced to expected communication.

We show that on one hand, a {\em single qubit} can sometimes replace an {\em infinite} stream of classical bits. 

\begin{theorem}[Separating one-way quantum vs classical]\label{thm:1-way-protocol} \hfill

    For the problem of auctioning two items to a single buyer, there is a Bayesian prior over independent item values, such that there is a revenue-optimal and IC one-way quantum auction protocol where the buyer sends 1 qubit and 2 classical bits; yet no finite classical auction protocol can achieve the optimal revenue.
\end{theorem}

Furthermore, interactive quantum protocols are even more powerful:

\begin{theorem}[Separating interactive quantum vs one-way quantum]\hfill\label{thm:interactive-vs-1-way}

    For the problem of auctioning two items to a single buyer, there is a Bayesian prior over independent item values, such that there is an IC revenue-optimal quantum auction protocol where the seller sends 1 qubit to the buyer, who replies with 1 qubit and 2 classical bits; yet no finite classical or one-way quantum auction protocol can achieve the optimal revenue.
\end{theorem}

However, in the worst case, even fully interactive quantum auction protocols cannot achieve optimal revenue in finite worst-case communication. 

\begin{theorem}[Limitations of finite interactive quantum auction protocols]\label{thm:non-algebraic}\hfill

    For the problem of auctioning two items to a single buyer, there is a Bayesian prior over independent item values, such that there is a revenue-optimal classical auction protocol that requires a constant number of bits {\em in exepctation}; yet no worst-case finite IC quantum auction protocol can achieve the optimal revenue.
\end{theorem}

\subsection{Key Takeaway: Thinking About Quantum and Incentives Together}\label{sub:quantumplusincentives}

Communication complexity in game-theoretic setting and quantum communication complexity have each been studied extensively over the past few decades, but in separate lines of work, and to a large extent in disjoint communities. The key takeaway from our work is that we can unlock significant advantages by thinking of both together - advantages that are not possible by composing disjoint results for classical-strategic communication and quantum-cooperative communication. 

In particular, it is important to note that our quantum speed-ups are not achievable by a generic quantum speed-up on communication (for example, Holevo's theorem states that an $n$-qubit quantum state, even with infinite precision, can only carry up to $n$ classical bits accessible information \cite{Holevo73}). In fact, they cannot be derived from a quantum speed-up on any non-strategic communication problem. 

We further note that our infinite separations require quantum operations with infinite precision. While this is clearly not practical, it is the standard textbook model of quantum computing, and, to the best of our knowledge, no previous work on quantum-cooperative communication complexity exhibits infinite separations\footnote{ Compared with quantum advantages in interactive proofs,  e.g.~the celebrated $\mathsf{MIP}^* = \mathsf{RE} $~\cite{JNVWY21}, it is worth noting that the latter is only an ``unbounded'' separation, i.e.~that separation still needs the complexity of the $ \mathsf{MIP}^*$ provers to go to infinity. In contrast, we show separation of 3 qubits vs infinitely many classical bits.}.

We highlight some of the ways in which our results are distinct from previous work in either line of work in Figure~\ref{fig:surprise-quantum} and Figure \ref{fig:surprise-AGT}.

\begin{figure}[t]
\centering
\begin{tcolorbox}

{\bf Surprises for the quantum side:}
\begin{itemize}
    \item \textbf{Exponential quantum speedup on a natural problem (Theorem~\ref{thm:combinatorial-protocol}).} For combinatorial and unit-demand auctions, we find quantum protocols that are exponentially more efficient than any classical protocol. Existing exponential quantum-classical separations (in cooperative communication complexity) are for problems like Hidden Matching Problem~\cite{BZJK04} that were designed for the purpose of exhibiting a separation. In contrast, the strategic communication problem we study here was considered before in classical algorithmic game theory~\cite{BabaioffGN22, RubinsteinZ21}.

    \item \textbf{Infinite 1-way vs. 1-way + entanglement separation (Theorem~\ref{thm:interactive-vs-1-way}). }%
    Specifically, we show that a one-way protocol with pre-shared entanglement (an EPR pair) is infinitely more efficient than any one-way protocol with no shared entanglement. This unique separation does not exist in the cooperative quantum communication environment as the honest sender can always prepare the EPR pair and send one half through the channel.
    
    \item \textbf{Infinite quantum-classical separation (Theorem~\ref{thm:1-way-protocol}).} We construct an example where we can implement the optimal auction with a one-way quantum communication protocol with $3$ qubits in the worst case. However, no worst-case finite classical protocol can implement it.

\end{itemize}

\end{tcolorbox}

\caption{Summary of most surprising aspects of our results from quantum perspective.}
\label{fig:surprise-quantum}
\end{figure}

\begin{figure}
\centering
\begin{tcolorbox}

{\bf Surprises for the AGT side:}
\begin{itemize}
    \item \textbf{Infinite 1-way vs. interactive separation (Theorem~\ref{thm:interactive-vs-1-way}).} 
    We show an example where an interactive two-way quantum protocol is infinitely better than any one-way quantum protocol. This unique separation does not exist in the classical setting. Since with a trusted party (the seller in our setting), any classical protocol can be ``flattened'' to a one-way protocol incurring an exponential overhead in the worst-case communication complexity, yet this overhead remains finite.
    \item \textbf{Exponential lower bound on CC $\times$ payment (Theorem~\ref{thm:large-payments}).} We prove that our exponentially more efficient quantum protocol in Theorem~\ref{thm:combinatorial-protocol} necessitates an exponentially large payment, by establishing an exponential lower bound on the product of payment and communication complexity. This characteristic is distinctive to the quantum setting, as the exponential lower bounds for classical protocols~\cite{RubinsteinZ21} apply even with arbitrarily large payments.
    
\end{itemize}

\end{tcolorbox}

\caption{Summary of most surprising aspects of our results from AGT perspective.}
\label{fig:surprise-AGT}
\end{figure}

\paragraph*{Core conceptual idea: Efficient, samplable and verifiable distribution encodings.}
At a high level, the communication protocol of Bayesian auction boils down to the following task: The seller has a set $\mathcal{S}$ consisting of valid distributions of auction outcomes (allocation and payment). The buyer then selects a distribution $D$ from $\mathcal{S}$ that maximizes his utility. Subsequently, the seller draws a sample from the chosen distribution $D$ and outputs it as the outcome. In general, the complexity of describing a distribution is exponentially higher than specifying an outcome. It is worth noting that this task is completely trivial in a fully cooperative environment as the buyer can simply draw the outcome himself and send only this outcome to the seller. A natural quantum solution to this problem is that the buyer can efficiently encode $D$ in state $\sum_{x} \sqrt{D_x} \ket{x}$. To draw an element from the distribution, the seller only needs to measure this quantum state in the computational basis. However, this simple approach has a caveat -- in general it is information theoretically impossible for the seller to verify that this quantum state indeed encodes a valid distribution in set $\mathcal{S}$. To overcome this issue, our paper proposes two solutions. The first one is a general spot-check method -- with a small probability, after receiving the quantum state, the seller asks the buyer to send the whole classical description of the distribution; the seller verifies that the classical description is valid and the quantum state is close enough to this classical description. By imposing an exponentially large penalty, we ensure that the buyer is always incentivized to prepare a quantum state corresponding to a valid distribution while keeping the {\em expected} communication complexity low. The other solution works for the worst case communication complexity under some specific assumptions on the set of valid distribution of auction outcomes. These assumptions are satisfied for example by the canonical classically-hard example of~\cite{DaskalakisDT17}, but not in general (see Theorem~\ref{thm:non-algebraic}).  By carefully tailoring the quantum protocol to the desired auction, we can ensure that the space of distributions that the buyer can encode  coincides with the desired valid space.

\subsection{Additional Related Works}\label{sub:related}
Our work extends a rich tradition of studying mechanism design and game theory under the lens of communication complexity --- including auctions~\cite{NisanS06,BlumrosenNS07,BabaioffBS13,DobzinskiV13,DobzinskiNO14,Dobzinski16a,Dobzinski16b,BravermanMW16,Assadi17,BravermanMW18,EzraFNTW19,AssadiKSW20,BabaioffGN22, WeinbergZ22}, and also stable matching~\cite{GonczarowskiNOR19}, voting rules~\cite{ConitzerS05,ProcacciaR06a,CaragiannisP11,ServiceA12a},
 fair division~\cite{BranzeiN19,PlautR19}, computation of equilibria~\cite{ConitzerS04,HartM10,RoghgardenW16,GoosR18,GanorS18,BabichenkoDN19,BabichenkoR20,GanorSP21,BabichenkoR22}, and interdomain routing~\cite{LSZ11}. In particular, the communication complexity of IC implementing a mechanism  vs that of (non-IC) computing the outcome was the focus of~\cite{FadelS09,DobzinskiR21,RubinsteinSTWZ21,DobzinskiRV22}.

We show exponential separations (and for worst-case complexity --- infinite separations) between quantum and classical communication complexity of auctions. Earlier works on separating the two measures (in non-strategic settings) include
general boolean functions (constructed for obtaining separations)~\cite{ Raz99, BZJK04, Gav08, GKKRdW09, RegevK11}, sampling~\cite{ASTAVW03, Montanaro19}, and very recently also linear regression~\cite{MontanaroS22, TLWXH+22}.

Our work is also related to works on quantum game theory --- including nonlocal games~\cite{CHSH69, CHTW04}, quantization of classical games~\cite{EWL99, Meyer00}, quantum equilibria~\cite{Deckelbaum11, WeiZ13}, and quantum interactive games~\cite{GutoskiW07}. In particular, quantum interactive strategies are also studied in quantum interactive proofs~\cite{Watrous03, BSW10, JJUW10, NW19, JNVWY21}.


\subsection{Organization}
We begin with a review of the quantum communication model and mechanism design in Section~\ref{sec:prelim-quantum} and \ref{sec:prelim-MD}, and then introduce our main model of quantum auction protocols in Section~\ref{sec:model}. We give a high-level overview of all our proofs in Section~\ref{sec:overview}. Part~\ref{part:combinatorial} brings our results for expected communication: Our quantum auction protocol for combinatorial valuations in Section~\ref{sec:eps-IC-protocol}, and our lower bound against quantum auction protocols with bounded maximum payment in Section~\ref{sec:small-payments}. Part~\ref{part:2-items} focuses on worst-case communication: we begin with preliminaries of optimal 2-item auctions in Section~\ref{sec:prelim-2-items}; in Section~\ref{sec:1-way-protocol} we construct an example separating one-way quantum auction protocols from finite classical; in Sections~\ref{sec:limit-1-way} and~\ref{sec:barely-interactive} we separate interactive quantum auction protocols from one-way; and finally in Section~\ref{sec:finite-round} we show that in general no finite quantum auction protocol can guarantee optimal revenue.

%% file: prelim-mechanism.tex
\section{Preliminaries II: Mechanism Design}\label{sec:prelim-MD}

We consider the mechanism (auction) design for selling $n$ indivisible items to a single risk-neutral buyer. A buyer has a \emph{type} (valuation function) $v: 2^{[n]} \to \mathbb{R}_{\ge 0} $ specifying his value for each bundle (subset). We use $X$ to denote the type set, which contains all possible types of the buyer. For our purposes, it is important to define the two simplest and most widely studied class valuations:
\begin{itemize}
    \item \textbf{Additive} If there exists a value of each item $v_1, \ldots, v_n$ such that $v(S) = \sum_{i \in S} v_i$.
    \item \textbf{Unit-demand} If there exists a value of each item $v_1, \ldots, v_n$ such that $v(S) = \max_{i \in S} v_i$.
\end{itemize}
Some of our results also hold for more general classes of valuations%
\footnote{See e.g.~\cite{LehmannLN01} for definitions; they are not necessary for understanding our paper.}, which satisfy the following hierarchy of increasing generality:
\[
\text{additive, unit-demand} \subset \text{gross-substitutes} \subset \text{submodular} \subset \text{XOS} \subset \text{subadditive} \subset \text{combinatorial}.
\]
In addition to satisfying these structures, valuations are usually assumed to be monotone; our results hold for both monotone and non-monotone valuations.

Without loss of generality, we consider direct mechanisms. The buyer reports a type $v' \in X$ to the mechanism, and the mechanism then allocates a (randomized) bundle to the buyer and charges the buyer a price.  $\mathcal{M}$ consists of two functions.
\begin{itemize}
    \item An allocation function $\mathcal{A}: X \to [0, 1]^{2^{[n]}}$ gives the probability of allocating each bundle to the buyer declares to have each possible type.
    \item A payment function $\mathcal{Q}: X \to \mathbb{R}_{\ge 0} $ gives the price the buyer needs to pay for each declared type of the buyer.
\end{itemize}

Let $D$ be a distribution over bundles. With a slight abuse of notation, we denote by $v(D) = \mathbb{E}_{S \sim D} v(S)$ the expected value of the buyer with type $v$.

We say a mechanism $\mathcal{M}=(\mathcal{A}, \mathcal{Q})$ is \emph{incentive compatible} (IC) if 
\[
v(\mathcal{A}(v)) - \mathcal{Q}(v) \ge  v(\mathcal{A}(v')) - \mathcal{Q}(v') \quad \forall v, v' \in X.
\]

We say a mechanism $\mathcal{M}=(\mathcal{A}, \mathcal{Q})$ is \emph{$\varepsilon$-incentive compatible} ($\varepsilon$-IC) if 
\[
v(\mathcal{A}(v)) - \mathcal{Q}(v) \ge  v(\mathcal{A}(v')) - \mathcal{Q}(v') - \varepsilon \quad \forall v, v' \in X.
\]


We say a mechanism $\mathcal{M}=(\mathcal{A}, \mathcal{Q})$ is \emph{individually rational} (IR) if 
\[
v(\mathcal{A}(v)) - \mathcal{Q}(v) \ge 0 \quad \forall v \in X.
\]

We say a mechanism $\mathcal{M}=(\mathcal{A}, \mathcal{Q})$ is \emph{$\varepsilon$-individually rational} ($\varepsilon$-IR) if 
\[
v(\mathcal{A}(v)) - \mathcal{Q}(v) \ge -\varepsilon \quad \forall v \in X.
\]

For a mechanism $\mathcal{M}=(\mathcal{A}, \mathcal{Q})$, we denote by $u: X \to \mathbb{R}$ the expected utility of the buyer when he truthfully reports the valuation function. It follows from the definition that $u(v) = v(\mathcal{A}(v)) - \mathcal{Q}(v).$

\paragraph*{Revenue Maximization}
In this paper, we primarily focus on the revenue-optimal Bayesian mechanism design. In this setting, the buyer knows his type $v$ for sure. However, the seller only knows the probability distribution over $X.$ Let $f: X \to \mathbb{R}$ be the probability density function of this distribution.

The goal of revenue-optimal Bayesian mechanism design is to find an IC and IR mechanism $\mathcal{M}=(\mathcal{A}, \mathcal{Q})$ that maximizes the revenue of the seller:
\[
Rev = \int_{X} \mathcal{Q}(v) f(v) \dd v.
\]

%% file: newmodel-auction.tex
\section{Quantum Communication Model with a Strategic Player}\label{sec:model}

We now introduce the main strategic communication model of this work, which formally defines the elements of a two-player quantum communication protocol subject to strategic manipulation.   In essence, when the length of the protocol is fixed, this model is equivalent to the one used in the literature of quantum games and quantum interactive proofs  (see e.g.~\cite{GutoskiW07}). The communication is between one strategic player (we call it the buyer), and a truthful player (we call it the seller). In this setup, the seller initially possesses $n$ qubits, while the buyer has 
 $m$ qubits, with $S$ representing the finite set of possible communication outcomes. Initially, the joint state is $\ket{0}^{\otimes (n + m)}$. The communication proceeds in rounds (or steps). Since we will cover protocols with infinite worst-case communication complexity (but bounded expected communication), we do not specify the total number of rounds in our model. In each round, one of the player performs a local operation on qubits in their hand and then sends \emph{one} qubit to the other player. Or you can alternatively think there is a one-qubit register shared by both players. We assume the seller takes the first round and then they alternate in the following rounds.

\paragraph*{The seller's operations.}
Without loss of generality, we assume the seller only performs general measurements in her rounds\footnote{If the seller simply want to apply a unitary $U$, she can do it by letting  $A^i_{\perp} = U$, and $A^i_{x} = 0$ for any $x \in S$.}. For round $i$, let $\{A^i_x\}_{x \in S \cup \{ \perp \}}$, such that $\sum_{x \in S \cup \{ \perp \}} (A^i_{x})^\dagger A^i_x = I_{2^n}$, be the seller's measurements. Let $h_i \in S \cup \{ \perp\}$ be the measurement outcome of round $i$. For each round $i$, if $h_i = \perp$ then communication continues, otherwise the seller terminates the communication and outputs $h_i$ as the outcome\footnote{By the principle of deferred measurement (see e.g. Chapter 4.4 of ~\cite{NielsenC01}), any non-terminating measurements outcomes can be removed by adding more qubits to the system. Therefore, it is wlog to only consider measurement with at most one non-terminiating outcome($\perp$) each round.}.

\paragraph*{The buyer's operations.}
In our model, only seller can terminate the protocol and output an outcome. Therefore, by the principle of deferred measurement (see e.g. Chapter 4.4 of ~\cite{NielsenC01}), we wlog assume the buyer makes no measurement\footnote{By the principle of deferred measurement, the buyer can always obtain the same outcome distribution by adding more qubits to his local memory and removing measurements. Since the buyer is allowed to choose an arbitrarily large memory size $m$ (we will discuss it later), it is wlog not to consider buyer's measurement.}, and his only local operation in round $i$ is a unitary $U_i$ with dimension $2^{m+1}$, for $i = 2, 4, \ldots$.

\paragraph*{The seller's strategy.}
A seller's strategy  includes the following elements:
\begin{itemize}
    \item Set of communication outcomes: $S$.
    \item The size of initial local quantum memory: $n$.
    \item General measurements: $\{A^i_{x} \}_{x \in S \cup \{ \perp \}} $, for $i = 1, 3, 5, \ldots$.
\end{itemize}

\paragraph*{The buyer's strategy.}
A buyer's strategy  $s^{\text{buy}}$ includes the following elements:
\begin{itemize}
    \item The size of initial local quantum memory: $m$.
    \item Unitary operators: $U_i$, for $i = 2, 4, 6, \ldots$.
\end{itemize}

\paragraph*{Quantum auction protocols.}

Our objective is to implement an auction using the strategic quantum communication model previously outlined. Our quantum auction protocols are a generalization of classical auction protocols defined in~\cite{RubinsteinZ21}. The classical auction protocols are also a special case of Bayesian incentive compatibility (BIC)-incentivizable binary dynamic mechanism (BDM) defined in~\cite{FadelS09}. For simplicity, here we only define the quantum analog for auctions, but exploring the quantum communication complexity of mechanisms more broadly is an interesting direction for future work. 

\begin{definition}[Quantum auction protocols] \label{def:quantumauction}
    A \emph{quantum auction protocol} $\mathcal{P}$ that sells $n$ items to a single buyer with type space $X$ consists of:
\begin{itemize}
    \item A seller's strategy $s^{\mathcal{P}}_{\text{seller}}$ such that the outcomes in the outcome set $S$ are in the form $(B, p)$, where $B \subseteq [n]$ is a subset of items and $p \in \mathbb{R}_{\ge 0}$ is the price. 
    \item A suggested strategy function $s^{\mathcal{P}}_*$ that maps each valuation in the type space $X$ to a buyer's strategy.
\end{itemize}
\end{definition}

Given a seller's strategy $s_{\text{seller}}$, and a buyer's strategy $s_{\text{buyer}}$ let $D(s_{\text{seller}}, s_{\text{buyer}})$ be the outcome distribution of the communication\footnote{Throughout the paper we only consider seller's strategies which guarantee termination within finite steps with probability $1$ for any buyer's strategy.}. Let $\mathcal{A}(s_{\text{seller}}, s_{\text{buyer}})$ be the marginal distribution of the first component of $D(s_{\text{seller}}, s_{\text{buyer}})$. So, $\mathcal{A}(s_{\text{seller}}, s_{\text{buyer}})$ is a distribution over subsets of $[n]$. Let $\mathcal{Q}(s_{\text{seller}}, s_{\text{buyer}})$ be the expected value of the second component (price) of $D(s_{\text{seller}}, s_{\text{buyer}})$.

With definitions above, we say a quantum auction protocol $\mathcal{P}$ \emph{implements} the mechanism \[\mathcal{M}^\mathcal{P} = \left ( \mathcal{A}(s^{\mathcal{P}}_{\text{seller}}, s^{\mathcal{P}}_{*}(\cdot)), \mathcal{Q}((s^{\mathcal{P}}_{\text{seller}}, s^{\mathcal{P}}_{*}(\cdot))) \right).\]

Further, we say a quantum auction protocol $\mathcal{P}$ is $\varepsilon$-IC if for any type $v \in X$, and any buyer's strategy $\hat{s}$, the following holds,
\[
v \left ( \mathcal{A}(s^{\mathcal{P}}_{\text{seller}}, s^{\mathcal{P}}_{*}(v)) \right) -  \mathcal{Q}(s^{\mathcal{P}}_{\text{seller}}, s^{\mathcal{P}}_{*}(v))  \ge v \left ( \mathcal{A}(s^{\mathcal{P}}_{\text{seller}}, \hat{s}) \right) - \mathcal{Q}(s^{\mathcal{P}}_{\text{seller}}, \hat{s}) - \varepsilon.
\]
By definition, the mechanism implemented by an $\varepsilon$-IC quantum auction protocol is an $\varepsilon$-IC mechanism.

We say quantum auction protocol $\mathcal{P}$ is $\varepsilon$-IR if for any type $v \in X$
the following holds,
\[
v \left ( \mathcal{A}(s^{\mathcal{P}}_{\text{seller}}, s^{\mathcal{P}}_{*}(v)) \right) -  \mathcal{Q}(s^{\mathcal{P}}_{\text{seller}}, s^{\mathcal{P}}_{*}(v))  \ge - \varepsilon.
\]
By definition, the mechanism implemented by an $\varepsilon$-IR quantum auction protocol is an $\varepsilon$-IR mechanism.

Exactly IC and IR quantum protocols are defined similarly.

The main goal of our paper is to find an ($\varepsilon$-)IC and ($\varepsilon$-)IR quantum auction protocol that implements the revenue-optimal mechanism.